\documentclass[journal]{IEEEtranTCOM}

\normalsize

\usepackage{amsmath,amsthm,amssymb,graphicx,tikz, bbold}
\usepackage[bookmarks=false]{hyperref}

\usepackage{amsfonts}
\usepackage{etoolbox}
\usepackage{mathtools}
\let\labelindent\relax
\usepackage{enumitem}
\usepackage{capt-of}
\usepackage{textcomp}
\usepackage{etoolbox}
\usepackage{array}
\usepackage{multirow}
\newcolumntype{P}[1]{>{\centering\arraybackslash}p{#1}}
\newcolumntype{M}[1]{>{\centering\arraybackslash}m{#1}}

\newtheorem{theorem}{Theorem}[section]
\newtheorem{lemma}[theorem]{Lemma}
\newtheorem{proposition}[theorem]{Proposition}

\newtheorem{corollary}[theorem]{Corollary}
\newtheorem{definition}[theorem]{Definition}

\newcounter{pfxc}[section]
\newenvironment{pfxc}[1][]{\refstepcounter{pfxc}\par\medskip
   \noindent \textbf{Prefix Constraint~\thepfxc. #1} \rmfamily}{\medskip}

\preto\subequations{\ifhmode\unskip\fi}

\renewcommand{\mod}{ \text{ mod }  }
\newcommand{\rvec}[1]{{\boldsymbol{\mathbf{\MakeLowercase{#1}}}}}

\newcommand{\rs}[1]{{\boldsymbol{\mathbf{\MakeLowercase{#1}}}}}
\newcommand{\dvec}[1]{{#1}}
\newcommand{\Cesaro}{Ces\'{a}ro }
\newcommand{\indep}{\perp \!\!\! \perp}
\newcommand{\cov}{\text{cov}}
\newcommand{\meig}{\rho_{\max}}
\newcommand{\tp}{\mathrm{T}}

\begin{document}
%
\title{Time-invariant Prefix Coding for LQG Control}
%
%
%

\author{Travis~Cuvelier,~\IEEEmembership{Student Member,~IEEE,} Takashi~Tanaka,~\IEEEmembership{Senior Member,~IEEE}
        and~Robert~W.~Heath~Jr.,~\IEEEmembership{Fellow,~IEEE,}
\thanks{T. Cuvelier is with the Department
of Electrical and Computer Engineering, The University of Texas at Austin,
TX, 78712 USA (e-mail: tcuvelier@utexas.edu).}
\thanks{T. Tanaka is with the Department
of Aerospace Engineering and Engineering Mechanics, The University of Texas at Austin,
TX, 78712 USA (e-mail: ttanaka@utexas.edu).}
\thanks{R. Heath is with the Department
of Electrical and Computer Engineering, North Carolina State University, Raleigh, 
NC, 27606 USA (e-mail: rwheathjr@ncsu.edu).}
\thanks{This paper has supplementary downloadable material available at http://ieeexplore.ieee.org., provided by the author. This material includes an appendix containing proofs of various results. Contact tcuvelier@utexas.edu for further questions about this work. The originally published version of the supplementary material included a proof that contained an error that turned out to be inconsequential. This updated preprint corrects this error, which originally appeared in Lemma A.7.  }}

%
%

\markboth{IEEE Journal on Selected Areas in Information Theory: ``Modern Compression"}%
{Submitted paper}
%



\maketitle

\begin{abstract}
Motivated by control with communication constraints, in this work we develop a time-invariant data compression architecture for linear–quadratic–Gaussian (LQG) control with minimum bitrate prefix-free feedback. For any fixed control performance, the approach we propose nearly achieves known directed information (DI) lower bounds on the time-average expected codeword length. We refine the analysis of a classical achievability approach, which required quantized plant measurements to be encoded via a time-varying lossless source code. We prove that the sequence of random variables describing the quantizations has a limiting distribution and that the quantizations may be encoded with a fixed source code optimized for this distribution without added time-asymptotic redundancy. Our result follows from analyzing the long-term stochastic behavior of the system, and permits us to additionally guarantee that the time-average codeword length (as opposed to expected length) is almost surely within a few bits of the minimum DI. To our knowledge, this time-invariant achievability result is the first in the literature. 
\end{abstract}


%
\IEEEpeerreviewmaketitle

\section{Introduction}
\IEEEPARstart{I}{n} this work we consider LQG control over communication networks. Our motivation is a scenario where measurements from a remote sensor platform are conveyed wirelessly to a controller. In such a system, the  bitrate of the feedback channel can be tied directly to the amount of physical layer resources (e.g., time, bandwidth, and power) that must be allocated to attain satisfactory control performance. Such resources are inherently scarce. This motivates approaches to control that minimize communication overhead; potentially enabling, for example, automated factories where many agents share the communication medium \cite{procSmartManufacturing}. 

We attack this problem via data compression; we develop quantizers and variable-length codecs for the LQG feedback link. We consider a setup where at each discrete timestep an encoder, co-located with a sensor that can fully observe the plant, conveys a variable-length packet of bits to a decoder co-located with the controller. We discuss various prefix constraints that can be imposed on the packets. Such constraints allow the decoder, and possibly other users sharing the same communication medium, to uniquely identify the end of the encoder's transmission. This can enable efficient resource sharing. The packet bitrate provides a notion of communication cost. We prove that for a fixed control performance, the approach we propose nearly achieves known lower bounds on the minimum achievable bitrate. We presently summarize our contribution.

\subsection{Our Contribution}
There have been several data compression architectures proposed in the prior literature for LQG control with near-minimum bitrate variable-length feedback. While several approaches are known to satisfy fixed constraints on the control cost with near-minimum bitrates, e.g. \cite{silvaFirst}\cite{tanakaISIT}\cite{kostinaTradeoff}, these approaches generally require that the output of a quantizer be losslessly encoded using a \textit{time-varying} source code; nominally a lossless code perfectly adapted to the probability distribution of the quantizer's output at every time $t$. In this work, we use tools from ergodic theory to demonstrate that an architecture based on that of \cite{tanakaISIT} can be used to achieve near minimum prefix-free bitrate LQG control with a completely time-invariant quantizer and prefix-free code design. As the prefix-free code used to encode the quantizer output is fixed, the scheme satisfies a well-motivated time-invariant prefix constraint that is significantly stronger than those considered in the prior art.  To our knowledge, this is the first such result in the literature.
\subsection{Literature Review}
This work considers minimum bitrate LQG control via dithered uniform quantization and variable length coding. An early paper to consider stabilizing a linear system with uniformly quantized feedback measurements was \cite{dfd1}. For a deterministic system, \cite{dfd1} analyzed the long-term behavior of the chaotic dynamics of the state vector using ergodic theory. The problem of stabilizing a Gauss--Markov plant over a feedback channel with a random, time-varying rate was considered in \cite{mineroStabilization}.  In the scalar case, a necessary and sufficient condition for stabilization was derived. In contrast, our work considers the problem of attaining a fixed control cost with variable-length coding (the number of bits to be transmitted at each time is chosen by the encoder, not by nature). This line of research follows from a model for LQG control with minimum rate variable-length feedback from \cite{silvaFirst}. For scalar plants, \cite{silvaFirst} lower bounded the time average expected bitrate of a prefix-free source codec used as an LQG feedback channel in terms of Massey's directed information (DI) \cite{masseyDI}. This motivated a rate-distortion problem for the tradeoff between (a lower bound on) communication cost, quantified by the DI, and LQG control performance. The rate distortion problem was solved for a restricted class of encoders in \cite{silvaFirst} and a more general class in \cite{silvaFirstPrime}. Using entropy-coded dithered uniform quantization (ECDQ), \cite{silvaFirst} and \cite{silvaFirstPrime} demonstrated that the DI lower bound was nearly achievable. ECDQ uses uniform quantizers and a sequence of independent, identically distributed (IID) uniform random variables shared between the encoder and decoder to effectively whiten the reconstruction error \cite{ditherQuant}. While under some assumptions (e.g. high quantizer resolutions and smooth source densities \cite{gish}), the reconstruction error in uniform quantization is approximately uniform over the quantizer cell and uncorrelated with the input, if a dither is used these hold exactly. Furthermore, ECDQ has an intuitive rate analysis. In \cite{silvaFirst} and \cite{silvaFirstPrime} it was assumed that at every timestep a quantized measurement is encoded using Shannon-Fano-Elias (SFE) prefix coding. In the codeword length analysis, it is assumed that the SFE codec used is designed optimally at each timestep for the conditional probability mass function (PMF) of the quantizer output given the dither realization. The proof of the near-achievability of the lower bounds then followed from \cite{ditherQuant}'s rate analysis.

The quantizer and source codec designs we propose follow from analyzing a DI/LQG cost rate-distortion function. While the DI-based bitrate lower bound in \cite{silvaFirst} purported to apply to systems using dithering, an error was discovered in \cite{milanCorrection}. Revised proofs in \cite{milanCorrection} (see also \cite{milanMainResult}) and \cite{ourConverseLetter} established that the DI lower bound on time average bitrate holds even when the encoder and decoder share randomness. The rate-distortion formulation of  \cite{silvaFirstPrime} was extended to MIMO plants in \cite{SDP_DI}. In particular, \cite{SDP_DI} analyzed the optimization over a randomized encoder and decoder policy space. This lead to a formulation of an optimal test channel consisting of an ``encoder" that conveys a linear/Gaussian plant measurement to a ``decoder/controller" consisting of a Kalman filter (KF) and certainty equivalent controller. The minimal DI attainable for any limit on LQG control performance was shown to be a convex (log-det) optimization. In \cite{kostinaTradeoff}, via \cite{verduVariableLength}, the DI lower bound for prefix-free codes was extended to the more general class of uniquely decodable codes. Analytical lower bounds on the DI cost as a function of control performance were also derived. The lower bounds in \cite{kostinaTradeoff} are applicable to plants with non-Gaussian process noise.  Our work is also related to nonanticipative rate distortion theory and its application to the causal tracking of Gauss/Markov sources (cf. \cite{ognrdf} and \cite{ognrdfprime}). In particular, a rate-distortion lower bound on the bitrate required to asymptotically estimate the state of an uncontrolled system is computed in \cite{photisSRDF_SIAM} and \cite{photisSRDF_STSP} via dynamic programming and reverse waterfilling. 

In \cite{tanakaISIT}, the achievability approach from \cite{silvaFirst} was extended to MIMO plants. In \cite{tanakaISIT},  linear measurements, dithered element-wise uniform quantization, KFs, and certainty equivalent control are used to develop a system where the feedback from plant to controller is discrete but with system variables with identical means and covariances to those in \cite{SDP_DI}'s optimal test channel. This ensures that the LQG performance is equivalent to that in the test channel, and leads to an asymptotic bound on the conditional entropy of the quantizer output (given the dither) within a few bits of the DI lower bound. This result proved that conveying the quantized measurements from the encoder to the decoder via a time-varying SFE codec that accounts for the dither asymptotically achieves a time average bitrate near the lower bound. Dithered quantization and time-varying entropy coding is likewise used in \cite{photisSRDF_STSP} to demonstrate the near-achievability of the respective lower bounds. An achievability approach not relying on dithered quantization was provided in \cite{kostinaTradeoff}. The approach in  \cite{kostinaTradeoff} uses lattice quantization and entropy coding. In particular, using a bound on the output entropy of a lattice quantizer from \cite{kostinaOutputEntropy}, \cite{kostinaTradeoff} demonstrates that the entropy of quantized innovations is close to a corresponding lower bound in the high rate/strict control cost regime. While the quantization/coding approaches in \cite{tanakaISIT}, \cite{kostinaTradeoff}, and \cite{photisSRDF_STSP} can be shown to nearly achieve respective lower bounds, they rely on time-varying lossless source codecs.  

The upper bounds on achievable rate in \cite{silvaFirst}, \cite{tanakaISIT} and \cite{kostinaTradeoff} are developed in terms of the output entropy of a quantizer. While a lossless codec can be used to encode the quantizations into a variable-length binary string without delay and with an expected length close to this entropy, this generally requires the codec to be adapted, at every timestep, to the probability distribution of the quantizer output. This complication is compounded in \cite{silvaFirst} and   \cite{tanakaISIT}, as the source codec must be adapted to the conditional probability distribution of the quantizer output given the dither.

Work on control with fixed-length feedback is also relevant. It is well established that a linear plant driven by unbounded process noise cannot be stabilized in the mean square sense with feedback that undergoes time-invariant, memoryless, fixed-length quantization \cite{nairevans}.  The problem of minimum bitrate stabilization with fixed-length feedback was considered in \cite{yukselStabilization}, \cite{r3_add2}, and \cite{r3_add3}. Stabilization via an adaptive (zooming) fixed-length quantizer was considered in \cite{yukselStabilization}. Using tools from ergodic theory, \cite{yukselStabilization} analyzed the long-term behavior of the state and quantizer parameters and proved the existence of limiting distributions. It is proven that a particular quantizer achieves finite control cost \cite{yukselStabilization}. In the present work, we will use similar theory to prove time-invariant achievability results for variable-length coding under a constraint on LQG cost. In \cite{r3_add2}, a theoretical analysis was conducted to determine the minimum necessary and sufficient fixed-length feedback bitrate required to stabilize an unstable scalar linear system driven by process noise with a bounded $\alpha$ moment. The minimum bitrate required to asymptotically stabilize the system in any moment $\beta < \alpha$ is shown to exceed the plant's autoregressive coefficient by at most one bit. This analysis unified special cases appearing in prior work. A fixed-length stabilization algorithm (a time-varying quantizer design) that achieves \cite{r3_add2}'s fundamental limit in the presence of unbounded process noise was proposed in  \cite{r3_add3}. In \cite{fixedLenCoding}, fixed-length quantizers were designed to  minimize control cost. Using a Lloyd-Max style quantizer designed at each timestep, an optimal greedy control policy was developed and exhibited competitive performance \cite{fixedLenCoding}. In our work we consider the less restrictive variable-length feedback setting. 

There is also relevant work pertaining to fixed and variable length strategies for joint-source channel coding. In an early work considering feedback over noisy communication channels, \cite{r3_add1} proposed to design fixed-length encoder and controller strategies to minimize LQG cost via an alternating optimization. Dynamic programming optimizations for the optimal controller given a fixed encoder, the optimal encoder given fixed controllers, and related structural results were derived. More recently \cite{r1_17_oster_cdc} developed a family of \textit{stabilizing codes} for stabilizing and controlling linear systems over a packetized erasure channel. Essentially, sequences of packet messages are designed such that performance guarantee holds given that some fraction of the packets arrive. 

In this work, we refine the analysis on the dithered quantizer output entropy from \cite{tanakaISIT}; restating classical results that reduce the space-filling gap and and bound the unconditioned output entropy of the quantizer. We use ergodic theory to analyze the long-term behavior of the system, and demonstrate that it is sufficient to encode the quantizer outputs using a fixed, time-invariant entropy code without incurring an appreciable increase in communication cost over \cite{tanakaISIT}. In particular, we use results from \cite{ito_invariant} to prove the existence of an invariant measure for the Markov chain that describes the quantizer's \textit{inputs}. We then use theorems from \cite{mcmcReview} to verify that the chain both converges to the invariant measure and has an ergodic property. Our proof of this measure's existence follows from an analysis of Lebesgue weakly transient sets, which, for Markov chains in Euclidean spaces, provide a necessary and sufficient condition for the existence of an invariant measure with a strictly positive probability density function (PDF). The convergence and ergodicity of the chain is more-or-less immediate via the verification of an irreducibility condition often encountered in the literature on Markov Chain Monte Carlo \cite{mcmcReview}. We propose to encode the quantizations using a fixed time-invariant SFE--style source codec designed for the quantizer output PMF induced the invariant measure. Our use of a fixed prefix code ensures that the system satisfies a stronger prefix constraint with respect to prior approaches. The ergodic property leads to a novel ``almost sure" guarantee on the system's time average codeword length (as opposed to time average \textit{expected} length). We then use basic information theoretic inequalities to demonstrate that the Kullback–Leibler (KL) divergence (relative entropy) between the true quantizer output at time $t$ and the output induced by the invariant measure tends to zero as $t\rightarrow\infty$. This recovers a guarantee on the time average \textit{expected} codeword length. After the initial submission of this work, we generalized our initial results on time-invariant achievability to a more general class of LQG control systems. This work's revision incorporates these generalizations, some of which appear in the conference proceedings \cite{meditcomm}.

Before concluding our discussion of the prior art, it worth mentioning that the mathematical machinery used to establish our main result (namely the proofs pertaining to the existence of the limiting distribution, its ergodic properties, and proof of the chain's convergence in the KL sense) are not the only relevant tools available. In particular, in \cite{yuk_feller} a generalization of the notion of Feller regular Markov kernels (cf. e.g. \cite{harMinicourse}) was introduced and used to study the invariance and convergence properties of various adaptive quantization schemes. In the context of quantized control over an erasure channel, the theory of petite sets was used in \cite{yukpetithr} to establish positive Harris recurrence for the general state space Markov chain describing the adapted quantizer bin size and the system state.  Such chains necessarily admit an invariant probability measure \cite{yukpetithr}. There is recent work relating a general state space Markov chain's convergence to an invariant measure in the total-variation sense to convergence in sense of KL divergence \cite{yukarxiv}. In \cite{yukpub}, this result is used to analyze the stochastic stability of nonlinear filters in controlled dynamical dynamical systems. A nonasymptotic analysis of the KL-sense convergence of Langevin Markov chain Monte Carlo was performed in \cite{mlr_langegin_mcmc} via viewing the diffusion as a gradient flow (path of steepest descent) in the space of probability measures. For completeness, in this work we provide a direct, simple proof via Shannon-type inequalities that the our quantizer's outputs converge in the KL sense to the relevant limiting distribution.  

\textbf{Notation and Organization:} Constant scalars and vectors are denoted by lower-case letters $x$. If ${x}$ is a vector, $[{x}]_{i}$ denotes its $i^{\mathrm{th}}$ element. For vectors let $\lVert x \rVert_{2}$ denote the Euclidean norm, and let $\lVert x \rVert_{\infty} = \max_{i} |[x]_{i}|$. Matrices are denoted by capital letters $X$, the identity matrix in $\mathbb{R}^{m\times m}$ by $I_{m}$, the $0$ vector in $\mathbb{R}^{m}$ by $0_{m}$, and the $0$ matrix in $\mathbb{R}^{m\times m}$ by $0_{m\times m}$. Let $\lVert X \rVert_{2}$ denote the largest singular value of $X$. Let $\meig(X)$ denote the spectral radius of $X$, namely $\meig(X) = \max |\lambda| \text{ s.t. } \lambda \text{ is an eigenvalue of } X$. We write P(S)D for ``symmetric positive (semi)definite", and let $\mathbb{S}^{m}_{+}$ denote the set of $m\times m$ PSD matrices. We let $\succ$, $\succeq$ denote the standard partial order on the PSD cone, e.g. if $A,B\in\mathbb{R}^{m}$, we write $A\succ B$ if $A-B$ is PD, likewise $A\succeq B$ if $A-B$ is PSD. Random scalars or vectors are written in boldface $\rvec{x}$. If $\rvec{a}$ is discrete, we write $\mathbb{P}_{\rvec{a}}[\rvec{a}=a]=\mathbb{P}_{\rvec{a}}[a]$, likewise for conditional PMFs. We write $\rs{a}\indep \rs{b}$ to denote that $\rs{a}$ and $\rs{b}$ are independent. We write $\rs{a}\overset{\mathrm{a.s.}}{=}\rs{b}$ if $\mathbb{P}_{\rs{a},\rs{b}}[\rs{a}=\rs{b}]=1$, and define $\overset{\mathrm{a.s.}}{\ge}\rs{b}$, $\overset{\mathrm{a.s.}}{<}\rs{b}$, etc. analogously. For $\rvec{x}$ a random vector, $\text{cov}(\rvec{x}) = \mathbb{E}[\rvec{x}\rvec{x}^{\mathrm{T}}]- \mathbb{E}[\rvec{x}]\mathbb{E}[\rvec{x}]^{\mathrm{T}}$. 
Denote the set of finite-length binary strings $\{0,1\}^*$. For time domain sequences, let $\{\rs{x}_{t}\}$ denote $(\rs{x}_{0},\rs{x}_{1},\dots)$, $\rs{x}_{a}^{b}=(\rs{x}_{a},\dots,\rs{x}_{b})$ if $b\ge a$, and $\rs{x}_{a}^{b}=\emptyset$ otherwise. We let $\rs{x}^{b}=  \rs{x}_{0}^{b}$. For a topological space $\mathbb{X}$, let $\mathbb{B}(\mathbb{X})$ denote the standard Borel $\sigma$-algebra of $\mathbb{X}$. For Euclidean spaces, let $\lambda$ denote the Lebesgue measure (e.g., if $\mathbb{X}$ is $\mathbb{R}^n$, then for $\mathcal{K}\in\mathbb{B}(\mathbb{X})$, $\lambda(\mathcal{K})$ is the volume of $\mathcal{K}$ in $\mathbb{R}^{n}$). For a set $\mathcal{K}$, define the indicator function of $\dvec{x}\in\mathcal{K}$ as $1_{\dvec{x}\in\mathcal{S}}$.  

In Section \ref{sec:systemmodel} we formulate the problem of LQG control with minimum rate prefix-free coding in the feedback link. Section \ref{sec:converse} restates the rate-distortion formulation and overviews the optimal test channel from \cite{SDP_DI}. Our main results are in Section \ref{sec:main}. We begin by overviewing the achievability approach and its key ingredients in Section \ref{ssec:allthekeys}. Section \ref{ssec:tiachiev} provides an overview of our time-invariant availability approach, together with a statement of our main result. We prove the main result in Section \ref{ssec:tiachievpf}, relegating the proofs of some lemmas to Appendix A in the online supplementary material. We conclude in Section \ref{sec:conclcusion}. 
\section{System Model and Problem Formulation}\label{sec:systemmodel}
We consider the system model depicted in Fig. \ref{fig:ditharch}. We consider a time-invariant MIMO plant controlled via a feedback model where communication occurs over an ideal (delay and error free) binary channel. The plant is fully observable to an encoder/sensor block, which conveys a variable-length binary codeword $\rvec{a}_{t}\in\{0,1\}^*$ over the channel to a combined decoder/controller. Upon receipt of the codeword, the decoder/controller designs the control input. Denote the state vector $\rvec{x}_t\in \mathbb{R}^{m}$, the control input $\rvec{u}_{t}\in\mathbb{R}^{u}$, and let $\rvec{w}_{t}\sim\mathcal{N}(\dvec{0}_{m},{W})$ denote processes noise assumed to be IID over time. We assume ${W}\succ{0}_{m\times m}$, i.e., the process noise covariance is full rank. We assume assume that $\rvec{x}_{0}\sim\mathcal{N}(\dvec{0},{X}_0)$ for some $X_0\succeq 0$.  For ${A}\in\mathbb{R}^{m\times m}$ the system matrix  and ${B}\in\mathbb{R}^{m\times u}$ the  feedback gain matrix, for $t\ge 0$ the plant dynamics are given by
\begin{align}\label{eq:ssmodel}
    \rvec{x}_{t+1} = {A}\rvec{x}_{t}+{B}\rvec{u}_{t}+\rvec{w}_{t}. 
\end{align} To ensure finite control cost is attainable, we assume $({A},{B})$ are stabilizable. 

For generality, we assume that the encoder/sensor and the decoder/controller may be randomized. In Fig. \ref{fig:ditharch}, we assume that the encoder/sensor and decoder/controller share access to a common random \textit{dither signal}, $\{\rs{\delta}_{t}\}$. The dither is assumed to be IID over time. In real-world systems, this \textit{shared randomness} can be effectively accomplished using two synchronized pseudorandom number generators at the encoder and decoder. The encoder/sensor policy in Fig. \ref{fig:ditharch} is a sequence of causally conditioned Borel measurable kernels denoted
\begin{align}\label{eq:encdithpol}
    \mathbb{P}_{\mathrm{E}}[\rvec{a}_{0}^{\infty}|| \rs{\delta}_{0}^{\infty},\rvec{x}_{0}^{\infty}] = \left\{ \mathbb{P}_{\mathrm{E},t}=\mathbb{P}_{\rvec{a}_{t}|\rvec{a}_{0}^{t-1},\rs{\delta}_{0}^{t},\rvec{x}_{0}^{t}}\right\}_t. 
\end{align} The corresponding decoder/controller policy is given by 
\begin{align}\label{eq:contdithpol}
    \mathbb{P}_{\mathrm{C}}[\rvec{u}_{0}^{\infty}|| \rvec{a}_{0}^{\infty},\rs{\delta}_{0}^{\infty}] = \left\{ \mathbb{P}_{\mathrm{C},t}=\mathbb{P}_{\rvec{u}_{t}|\rvec{a}_{0}^{t},\rs{\delta}_{0}^{t},\rvec{u}_{0}^{t-1}}\right\}_t.
\end{align} Note that under the dynamics (\ref{eq:ssmodel}), $\rs{x}_{0}^{t}$ is a deterministic function of $\rvec{x}_{0}$, $\rvec{a}_{0}^{t-1}$, $\rvec{u}_{0}^{t-1}$, and $\rvec{w}_{0}^{t-1}$. We enforce conditional independence assumptions in the system model by a factorization of the one-step transition kernels for $\rvec{a}_{t}$, $\rs{\delta}_{t}$, $\rvec{u}_{t}$, and  $\rvec{w}_{t}$. The assumed conditional independence relationships induced between the system variables are illustrated through the factorizations of the transition kernels in (\ref{eq:ditherFactorization}) at the top of the following page, and are  discussed in Fig. \ref{fig:ditharch}. For $\mathcal{A}, \mathcal{D},\mathcal{U},\mathcal{W}$ measurable subsets, for  $t\ge 0$, we assume the transition kernels factorize via (\ref{eq:ditherFactorizationtneq0}). The conditional measure of $(\rvec{a}_0,\rvec{\delta}_0,\rvec{u}_{0},\rvec{w}_{0})$ given $\rvec{x}_{0}$ is given in (\ref{eq:ditherFactorizationteq0}).
\begin{figure*}[!t]
\normalsize
\begin{subequations}\label{eq:ditherFactorization} 
\begin{multline}
    \mathbb{P}[(\rvec{a}_{t+1}\in \mathcal{A})\cap(\rs{\delta}_{t+1}\in\mathcal{D})\cap(\rvec{u}_{t+1}\in\mathcal{U})\cap(\rvec{w}_{t+1}\in\mathcal{W})|\rvec{a}_{0}^{t},\rs{\delta}_{0}^{t},\rvec{u}_{0}^{t},\rvec{w}_{0}^{t},\rvec{x}_{0}] =\\ \mathbb{P}_{\mathrm{E},t+1}[\rvec{a}_{t+1}\in\mathcal{A}|\rvec{a}_{0}^{t},\rs{\delta}_{0}^{t+1},\rvec{x}_{0}^{t+1}]\mathbb{P}_{\mathrm{C},t+1}[\rvec{u}_{t+1}\in\mathcal{U}|\rvec{a}_{0}^{t+1},\rs{\delta}_{0}^{t+1},\rvec{u}_{0}^{t}]\mathbb{P}[\rs{\delta}_{t+1}\in\mathcal{D}]\mathbb{P}[\rvec{w}_{t+1}\in\mathcal{W}],\text{ }t\ge 0\label{eq:ditherFactorizationtneq0},
\end{multline} 
\begin{align}
\mathbb{P}[(\rvec{a}_{0}\in \mathcal{A})\cap(\rs{\delta}_{0}\in\mathcal{D})\cap(\rvec{u}_{0}\in\mathcal{U})\cap(\rvec{w}_{0}\in\mathcal{W})|\rvec{x}_{0}] = \mathbb{P}[\rs{\delta}_{0}\in\mathcal{D}]\mathbb{P}_{\mathrm{E},0}[\rvec{a}_{0}\in\mathcal{A}|\rvec{x}_{0},\rs{\delta}_{0}]\mathbb{P}_{\mathrm{C},0}[\rvec{u}_{0}\in\mathcal{U}|\rvec{a}_{0},\rs{\delta}_0]\mathbb{P}[\rvec{w}_{0}\in\mathcal{W}]\label{eq:ditherFactorizationteq0}\end{align}\end{subequations}
\hrulefill
\vspace*{4pt}
\end{figure*}

\begin{figure}
	\centering
	\includegraphics[scale = .23]{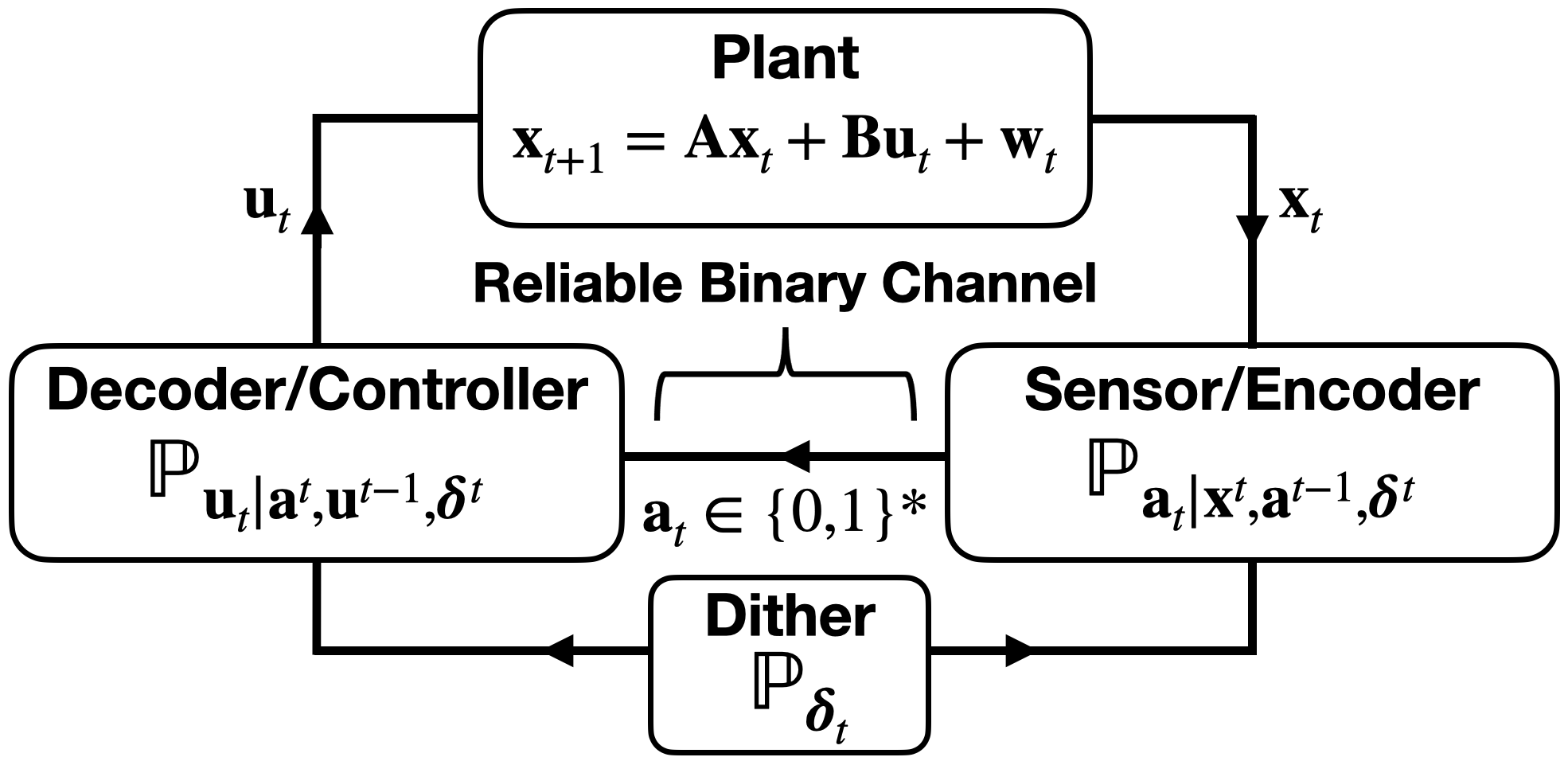}
	\caption{The system model with dithering. The encoder policy allows the codeword $\rs{a}_{t}$ to be generated randomly given ``all the information known to the encoder at time $t$". When $\rs{a}_{t}$ arrives at the decoder, the decoder can randomly generate its control input given $\rs{a}_{t}$ as well as its previous knowledge. Notably, both the encoder and decoder share access to $\rs{\delta}_{t}$, an IID sequence generated ``independently" of all past system variables.  }\label{fig:ditharch}
\end{figure}

The length of the codewords $\{\rs{a}_{t}\}$ quantifies the communication cost. For a codeword $\rvec{a}_{t}\in\{0,1\}^{*}$, denote its length in bits by $\ell(\rvec{a}_{t})$. The problem of interest is to minimize the time average expected bitrate subject to a constraint on control performance, quantified via the standard LQG cost. We will impose prefix constraints on the codewords $\{\rvec{a}_{t}\}$. These constraints will allow the decoder (and possibly other agents sharing the same communication medium) to uniquely identify the end of the transmission from the encoder. Three possible prefix constraints are:
 \begin{pfxc}\label{pfxc:pf1}
    For any realizations ($\rs{a}_{0}^{t-1}={a}_{0}^{t-1},\rs{\delta}_{0}^{t}={\delta}_{0}^{t},\rs{u}_{0}^{t-1}={u}_{0}^{t-1}$), for all distinct $a_1,a_2\in\{0,1\}^*$ with $\mathbb{P}_{\rs{a}_t|\rs{a}_{0}^{t-1},\rs{\delta}_{0}^{t},\rs{u}_{0}^{t-1}}[\rs{a}_t=a_1|\rs{a}_{0}^{t-1}={a}_{0}^{t-1},\rs{\delta}_{0}^{t}={\delta}_{0}^{t},\rs{u}_{0}^{t-1}={u}_{0}^{t-1}]>0$ and $\mathbb{P}_{\rs{a}_t|\rs{a}_{0}^{t-1},\rs{\delta}_{0}^{t},\rs{u}_{0}^{t-1}}[\rs{a}_t=a_2|\rs{a}_{0}^{t-1}={a}_{0}^{t-1},\rs{\delta}_{0}^{t}={\delta}_{0}^{t},\rs{u}_{0}^{t-1}={u}_{0}^{t-1}]>0$, $a_1$ is not a prefix of $a_2$.
\end{pfxc} 
\begin{pfxc}\label{pfxc:pf2}
    For all distinct $a_1,a_2\in\{0,1\}^*$ with $\mathbb{P}_{\rs{a}_t}[\rs{a}_t=a_1]>0$ and $\mathbb{P}_{\rs{a}_t}[\rs{a}_t=a_2]>0$, $a_1$ is not a prefix of $a_2$.
\end{pfxc} 
\begin{pfxc}\label{pfxc:pf3}
    For all $i,j$ and distinct $a_1,a_2\in\{0,1\}^*$ with $\mathbb{P}_{\rs{a}_i}[\rs{a}_i=a_1]>0$ and $\mathbb{P}_{\rs{a}_j}[\rs{a}_j=a_2]>0$, $a_1$ is not a prefix of $a_2$.
\end{pfxc} 

\noindent Prefix Constraints \ref{pfxc:pf1} and \ref{pfxc:pf2} were defined in \cite{ourConverseLetter}. Constraint \ref{pfxc:pf1} is the least strict. It allows any agent with knowledge of the information possessed by the decoder at time $t$ to uniquely identify the end of the encoder's transmission at time $t$. A downside, however, is that this information may be \textit{necessary} to determine the end of the codeword. This complicates the system architecture and may inhibit other agents from recognizing the end of the transmission. Constraint \ref{pfxc:pf2} is notionally stricter; it guarantees that any agent who knows the \textit{codebook} used by the encoder at time $t$ (precisely, the set  $\{b\in\{0,1\}^*\text{ } :\text{ }\mathbb{P}_{\rs{a}_t}[\rs{a}_t=b]>0\}$) can uniquely identify the end of the transmission. Under Constraint \ref{pfxc:pf2}, agents on the same network can identify the end of the transmission \textit{without} knowing $(\rs{a}_{0}^{t-1},\rs{\delta}_{0}^{t},\rs{u}_{0}^{t-1})$. Constraint \ref{pfxc:pf3} is a \textit{time-invariant} version of Constraint \ref{pfxc:pf2}. Constraint \ref{pfxc:pf3}  requires that the prefix condition holds across time, ensuring that any codeword used at time $t$ is not a prefix of any codeword used at time $t+m$ for any $m$. Any user with knowledge of the set $\{b\in\{0,1\}^*\text{ } :\text{ }\exists \text{ }t\text{ s.t. }\mathbb{P}_{\rs{a}_t}[\rs{a}_t=b]>0\}$ can uniquely identify the end of the transmission at any time $t$. Notably, to identify the end of the transmission, a user need not know the codebook used at time $t$, but only the  strings lying in the union of codebooks across time. Note that Constraint \ref{pfxc:pf3} is satisfied if the same prefix-free code is used for all $t$. 

We are interested in the optimization, for codewords conforming to Prefix Constraints \ref{pfxc:pf1}--\ref{pfxc:pf3}:  
\begin{equation}\label{eq:codewordLenghtOptimization}
\begin{aligned}
& \underset{\mathbb{P}_\mathrm{E}, \mathbb{P}_{\mathrm{C}}}{\inf}  \underset{T\rightarrow \infty}{\lim\sup}\text{ }\frac{1}{T+1}\sum\nolimits_{t=0}^{T}\mathbb{E}[\ell(\rvec{a}_{t})]\\  &\text{s.t. }   \underset{T\rightarrow \infty}{\lim\sup}\frac{1}{T+1}\sum\nolimits_{t=0}^{T}\mathbb{E}[\lVert \rvec{x}_{t+1} \rVert_{{Q}}^{2} +\lVert \rvec{u}_{t} \rVert_{\Phi}^{2}] \le \gamma,
\end{aligned} 
\end{equation} where ${Q}\succeq {0}$, $\Phi\succ {0}_{m\times m}$, and $\gamma$ is the maximum tolerable LQG cost. The minimization is over admissible sensor/encoder and decoder/controller policies described by 
(\ref{eq:encdithpol}) and (\ref{eq:contdithpol}). In Section \ref{sec:converse}, we discuss a lower bound on (\ref{eq:codewordLenghtOptimization}) that applies to all encoder and decoder policies conforming to (\ref{eq:ditherFactorization}) and any of the Prefix Constraints \ref{pfxc:pf1}--\ref{pfxc:pf3}. These bounds follow from \cite{ourConverseLetter}. Note that Constraint \ref{pfxc:pf1} was the notion of prefix-free considered in \cite{silvaFirst} and \cite{tanakaISIT}, while the ``prefix-free" version of the approach in \cite{kostinaTradeoff} conforms to Constraint \ref{pfxc:pf2}. To our knowledge, no variable-length compression architecture for LQG control in the prior work is known to both satisfy Constraint \ref{pfxc:pf3} and also achieve a codeword length provably close to any known lower bound on the optimization in (\ref{eq:codewordLenghtOptimization}). 

\section{Rate Distortion Lower Bound}\label{sec:converse}
We summarize the relevant results from \cite{ourConverseLetter} and \cite{SDP_DI} into the following theorem. 
\begin{theorem}\label{thm:converse}
Let the minimum communication cost attained by the optimization in (\ref{eq:codewordLenghtOptimization}) for an LQG cost constraint $\gamma$ be denoted $\mathcal{L}(\gamma)$. Let $S$ be a stabilizing solution to the discrete algebraic Riccati equation (DARE) $A^{\mathrm{T}}SA-S-A^{\mathrm{T}}SB(B^{\mathrm{T}}SB+\Phi)^{-1}B^{\mathrm{T}}SA+Q = 0_{m\times m}$, $K=-(B^{\mathrm{T}}SB+\Phi)^{-1}B^{\mathrm{T}}SA$, and $\Theta = K^{\mathrm{T}}(B^{\mathrm{T}}SB+\Phi)K$. Define the convex log-det  optimization
\begin{align}\label{eq:threestageRDF}
    \mathcal{R}(\gamma) &=& \left\{ \begin{aligned}
& \underset{\substack{P,\Pi, \in\mathbb{R}^{m\times m}\\P,\Pi\succeq 0_{m\times m}} }{\inf} \frac{1}{2}(-\log_{2}{\det{\Pi}}+\log_{2}{\det{W}} )\\ &\text{ }\text{s.t. }  \mathrm{Tr}(\Theta P)+\mathrm{Tr}(WS)\le \gamma\text{,  }\\&\text{ }\text{ }P\preceq APA^\mathrm{T}+W\text{, } \\&\text{ }\text{ }\text{ }\text{ }\begin{bmatrix}P-\Pi & PA^{\mathrm{T}} \\ AP & APA^{\mathrm{T}}+W \end{bmatrix}\succeq 0_{2m\times 2m}.
\end{aligned}\right. 
\end{align}  For a system conforming to Fig. \ref{fig:ditharch}, (\ref{eq:ditherFactorization}), and any of Prefix Constraints \ref{pfxc:pf1}--\ref{pfxc:pf3} we have
  $\mathcal{L}(\gamma)\ge\mathcal{R}(\gamma)$.  
\end{theorem} The proof of Theorem \ref{thm:converse} is immediate from \cite{ourConverseLetter} given that Constraint \ref{pfxc:pf3} is more stringent than Constraint \ref{pfxc:pf2}. The interpretation of the optimization in (\ref{eq:threestageRDF}) is aided by the three-stage test channel illustrated in Fig. \ref{fig:threeStageTI}. 
\begin{figure}
	\centering
	\includegraphics[scale = .21]{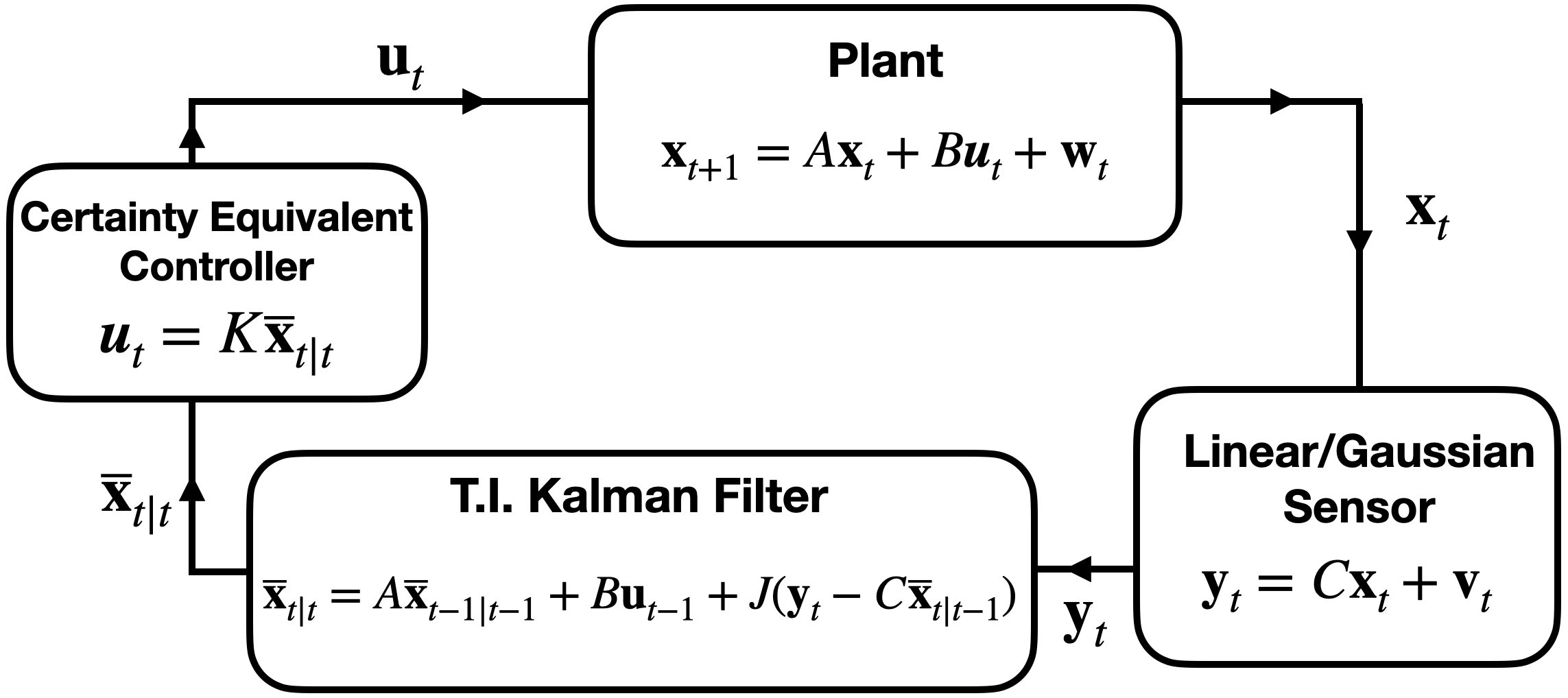}
    
	\caption{  The time-invariant three-stage test channel does not conform to the system model in Fig. \ref{fig:ditharch}, but will be used to analyze the approaches we propose.  }\label{fig:threeStageTI}
\end{figure} The test channel consists of an ``encoder" that conveys a linear/Gaussian plant measurement to a ``decoder"/controller. The decoder has a time-invariant KF to track the state, followed by a standard certainty equivalent controller. Denote the minimizing $P$ from (\ref{eq:threestageRDF}) by $\hat{P}$. Let $C\in \mathbb{R}^{m\times m}$ and $V\in \mathbb{R}^{m\times m}$, $V\succ 0_{m\times m}$ be any such matrices that satisfy
\begin{align}\label{eq:cdef}
    \hat{P}^{-1}-(A\hat{P}A^{\mathrm{T}}+W)^{-1}-C^{\mathrm{T}}V^{-1}C=0_{m\times m}.
\end{align}
The decoder receives the measurement $\rs{y}_{t} = C\rs{x}_{t}+\rs{v}_{t}$ where $\rs{v}_{t}\sim\mathcal{N}(0_{m},V)$ \text{ IID } and $\rs{v}_{t}\indep \rs{x}_{0}^{t}$. Let $\hat{P}_{+}= A\hat{P}A^{\mathrm{T}}+W$ and let $J= \hat{P}_{+}{C}^{\mathrm{T}}({C}\hat{P}_{+}{C}^{\mathrm{T}}+{V})^{-1}.$ Denote the filter's sequence of prior and posterior state estimates as $\{\rs{\overline{x}}_{t|t-1}\}$ and $\{\rs{\overline{x}}_{t|t}\}$. Let $\rs{\overline{x}}_{0|-1} = 0$. The filtering recursion is $\rs{\overline{x}}_{t|t} = \rs{\overline{x}}_{t|t-1}+J(\rs{y}_{t}-C\rs{\overline{x}}_{t|t-1})$ and $\rs{\overline{x}}_{t|t-1} = {A}\rs{\overline{x}}_{t-1|t-1}+{B}\rs{u}_{t-1}$. Define the prior and posterior error processes and their respective covariances via  $\rs{e}_{t} = \rs{x}_{t} - \rs{\overline{x}}_{t|t-1}$, $\overline{P}_{t|t-1} = \mathbb{E}[\rs{e}_{t}\rs{e}_{t}^{\mathrm{T}}]$ and $\rs{e}_{t|t} = \rs{x}_{t} - \rs{\overline{x}}_{t|t}$, $\overline{P}_{t|t} = \mathbb{E}[\rs{e}_{t|t}\rs{e}_{t|t}^{\mathrm{T}}]$. Note that for all $t\ge 0$, $\mathbb{E}[\rs{e}_{t}] = 0$ and  $\mathbb{E}[\rs{e}_{t|t}] = 0$. When $W\succ 0$, a discrete Lyaponov equation can be used to establish that for any $C$ satisfying (\ref{eq:cdef}), $({A},{C})$ is detectable; see \cite[below (25)]{tanakaSDRGM} for a similar argument. Since $W\succ 0_{m\times m}$, $(A,W^{\frac{1}{2}})$ is stabilizable. 
This implies that $ \lim_{t\rightarrow\infty} \overline{P}_{t|t-1} = \hat{P}_{+}$ and $\lim_{t\rightarrow\infty} \overline{P}_{t|t} = \hat{P}$ \cite{RDE_convergence}. Recall $K=-(B^{\mathrm{T}}SB+R)^{-1}B^{\mathrm{T}}SA$. The control input at time $t$ given by $\rs{u}_{t} = K\rs{\hat{x}}_{t|t}$.  It can be shown (see \cite{SDP_DI}) that, in the architecture of Fig. \ref{fig:threeStageTI} the control cost satisfies
\begin{IEEEeqnarray}{rCl}
    \underset{T\rightarrow \infty}{\lim}\frac{\sum\nolimits_{t=0}^{T}\mathbb{E}[\lVert \rvec{x}_{t+1} \rVert_{{Q}}^{2} +\lVert \rvec{u}_{t} \rVert_{{R}}^{2}]}{T+1} &=& \mathrm{Tr}(SW)+\mathrm{Tr}(\hat{P}\Theta)\label{eq:cntrcostxpres}\nonumber\\&\le& \gamma\label{eq:ref_threestageRDF}, 
\end{IEEEeqnarray} where (\ref{eq:ref_threestageRDF}) follows as $\hat{P}$ is a feasible solution of (\ref{eq:threestageRDF}). The minimum of (\ref{eq:threestageRDF}) is given by (see \cite{SDP_DI})
\begin{IEEEeqnarray}{rCl}\label{eq:commcostxpres}
   \mathcal{R}(\gamma) &=&\frac{1}{2}\log_{2}{\frac{\det{\hat{P}_{+}}}{\det{\hat{P}}}}. 
\end{IEEEeqnarray}
We reiterate that (\ref{eq:commcostxpres}) lower bounds the communication cost attainable in the (original) architecture in Fig. \ref{fig:ditharch}. We will use (\ref{eq:cntrcostxpres}), (\ref{eq:commcostxpres}), and the test channel in the following section on achievability.

\section{Upper Bounds (Achievability)}\label{sec:main}
In this section, we present theoretical results demonstrating that, assuming access to a uniform dither signal in the architecture of Fig. \ref{fig:ditharch}, uniform (dithered) quantization coupled with time-invariant prefix-free source coding strategies can be used to achieve nearly optimal communication bitrate (with respect to the DI lower bound in (\ref{eq:threestageRDF})). We propose one approach where quantizations are encoded conditioned on the realization of the dither, but without any other time adaptation. This approach conforms to Prefix Constraint \ref{pfxc:pf1}, and is shown to achieve the same communication cost as the architecture in \cite{tanakaISIT}. We then propose an approach where quantization is performed with dither, but the encoding of the discrete quantizations into codewords is done without regard to the dither realization. This leads to a completely time-invariant approach; the same prefix-free codec is used to encode the quantizations at all time. This latter approach conforms to Prefix Constraint \ref{pfxc:pf3}, and achieves a bitrate at most one-bit-per-plant-dimension worse that the time-varying approach in \cite{tanakaISIT}. 

Fig. \ref{fig:overviewAchiev} illustrates an overview of the framework we will use to demonstrate achievability in this section. The approach conforms to the architecture in Fig. \ref{fig:ditharch} with the dither signal chosen as an IID seqeuence of element-wise mutually independent uniform random vectors. At a high level, at every time $t$,  encoder produces a particular linear measurement of the plant, which it then quantizes into a discrete random variable (a quantization), $\rvec{q}_{t}$, using an \textit{elementwise uniform quantizer with subtractive dither}. Each element of the dither sequence $\rvec{\delta}_{t}$ has IID elements with $[\rvec{\delta}_{t}]_{i}$ uniform on $[\frac{-\Delta}{2},\frac{\Delta}{2}]$. The encoder then encodes $\rvec{q}_{t}$ into a codeword, $\rvec{a}_{t}$, using a lossless \textit{Shannon-Fano-Elias (SFE) prefix-free code}. The decoder recovers $\rvec{q}_{t}$ exactly, and then designs the control input $\rvec{u}_{t}$ using $\rvec{q}_{t}$, $\rvec{\delta}_{t}$, and a previous KF estimate. In the next subsection, we describe SFE codes and dithered uniform quantization in the detail necessary to proceed with our analysis. 
\begin{figure}[h]
	\centering
	\includegraphics[scale = .225]{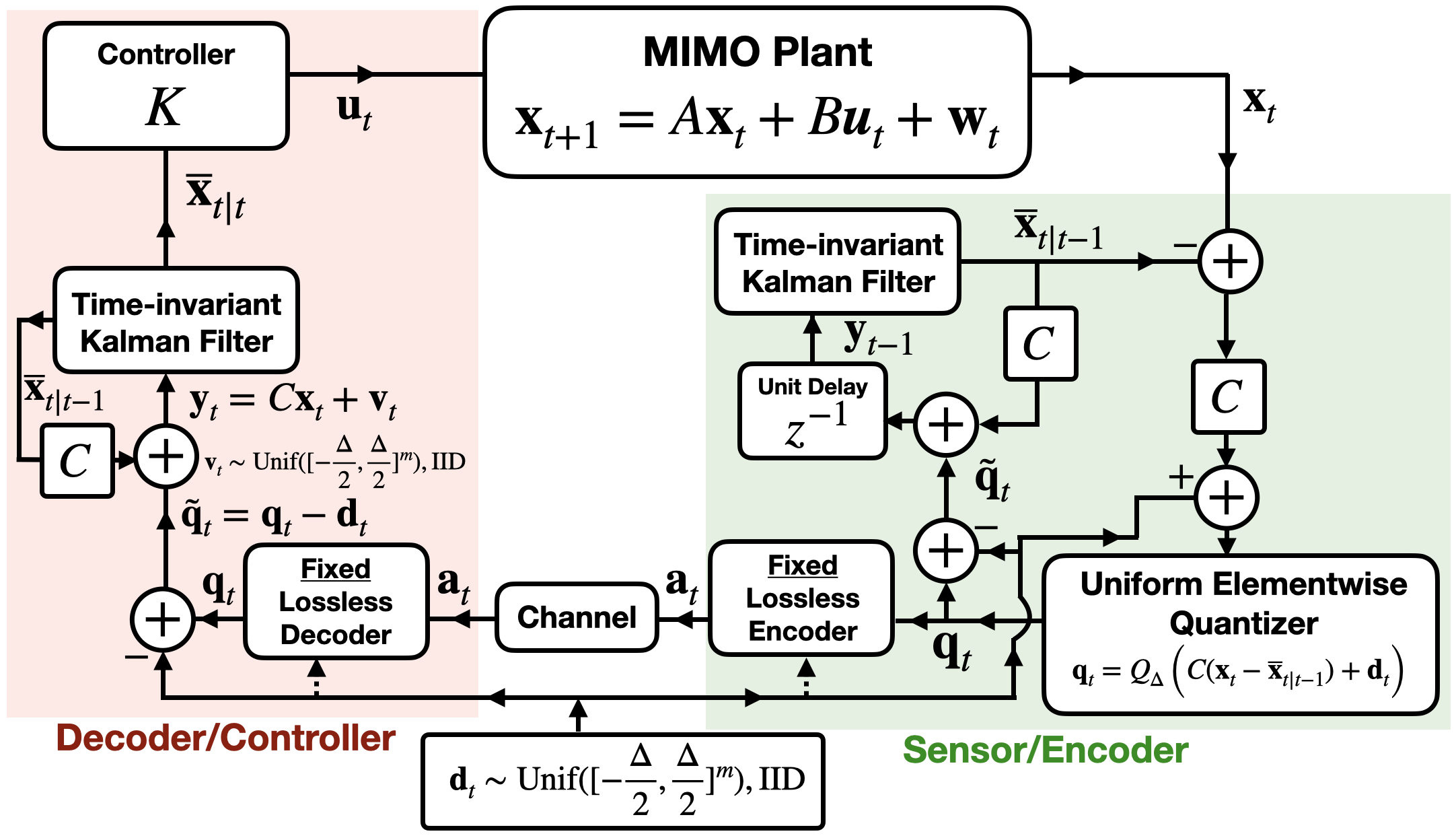}

	\caption{The achievability architecture. The dither sequence $\{\rs{\delta}_{t}\}$ are shared random vectors that are IID uniform on $[-\Delta/2,\Delta/2]^m$. The dither realization may be used by the entropy codec, but need not be (we consider both cases). }\label{fig:overviewAchiev}
\end{figure}
\subsection{Key ingredients}\label{ssec:allthekeys}
\subsubsection{Shannon-Fano-Elias codes \cite{elemIT}}\label{ssec:codec_key}
In this section we briefly outline the SFE approach to prefix-free source coding. We will pursue a general treatment, but will specialize the results to the quantization architecture in Fig. \ref{fig:overviewAchiev}. 

Let $\rs{q}$ denote a discrete random variable with (countable) range $\mathbb{A}$. Without loss of generality, it can be assumed that $\mathbb{A}=\mathbb{N}$ (if the alphabet is countably infinite) or $\mathbb{A}=\{0,1,\dots,r\}$, and that $\mathbb{P}_{\rs{q}}[q] > 0$ $\forall$ $q\in\mathbb{A}$. Let $\rvec{\delta}$ be a random variable on support $\mathbb{X}$ assumed to be known to both the encoder and decoder. Consider the problem of encoding $\rs{q}$ into a prefix-free codeword, such that it can be recovered at a decoder. In this scenario, we view $\rvec{q}$ as a quantization, and $\rvec{\delta}$ as shared randomness, akin to the dither sequence. Let $\overline{F}_{\rs{q}|\rs{\delta}}(q|\delta) =  \mathbb{P}_{\rs{q}}[\rs{q}<q|\delta]+\mathbb{P}_{\rs{q}|\rs{\delta}}[q|\delta]/2$. Define what we will refer to as the ``unsorted, conditional" encoding $\mathrm{C}^{\mathrm{U}}_{\rs{q}|\rs{\delta}}: \mathbb{A}\otimes \mathbb{X}\rightarrow \{0,1\}^{*}$ as 
\begin{multline}\label{eq:unsortedRecipeSI}
    \mathrm{C}_{\rs{q}|\rs{\delta}}^{\mathrm{U}}(q|\delta) = \left( \text{the binary expansion of } \overline{F}_{\rs{q}|\rs{\delta}}(q|\delta)\right.\\\left.\text{truncated to }\lceil-\log_{2}(\mathbb{P}_{\rs{q}|\rs{\delta}}[q|\delta])\rceil+1\text{ bits. }\right)
\end{multline} It can be shown that for any realization $\rvec{\delta}=\delta$ and ${q}_{1},{q}_{2}\in\mathbb{A}$ with $\mathbb{P}_{\rvec{q}|\rvec{\delta}}[q_{1}|\delta],\mathbb{P}_{\rvec{q}|\rvec{\delta}}[q_{2}|\delta]>0$ (e.g. any two quantizations $q_{1}$ and $q_{2}$ with nonzero probability of occurring given $\rvec{\delta}=\delta$), $\mathrm{C}_{\rs{q}|\rs{\delta}}^{\mathrm{U}}(q_1|\delta)$ is not a prefix of $\mathrm{C}_{\rs{q}|\rs{\delta}}^{\mathrm{U}}(q_2|\delta)$ and vice versa \cite[Chapter 5.9]{elemIT}. This property mirrors Prefix Constraint \ref{pfxc:pf1}, e.g. codewords are prefix-free \textit{given} the knowledge shared by the encoder and decoder. If $\mathrm{C}^{\mathrm{U}}_{\rs{q}|\rs{\delta}}$ is used to encode $\rvec{q}$ (given the realization of $\rvec{\delta}$), then the codeword length satisfies
\begin{align}\label{eq:unsortedsilength}
H(\rs{q}|\rs{\delta})    \le \mathbb{E}_{\rs{q},\rs{\delta}}[\mathrm{C}^{\mathrm{U}}_{\rs{q}|\rs{\delta}}(\rs{q}|\rs{\delta})] \le H(\rs{q}|\rs{\delta})+2.
\end{align}
We now state a construction that achieves a stronger prefix constraint. Define 
$\overline{F}_{\rs{q}}(q) =  \mathbb{P}_{\rs{q}}[\rs{q}<q]+\mathbb{P}_{\rs{q}}[q]/2$, and define the ``unsorted, unconditional" encoding function $\mathrm{C}^{\mathrm{U}}_{\rs{q}}: \mathbb{A}\rightarrow \{0,1\}^{*}$ as 
\begin{multline}\label{eq:unsorted}
    \mathrm{C}_{\rs{q}}^{\mathrm{U}}(q) = \left(\text{the binary expansion of } \overline{F}_{\rs{q}}(q)\text{ truncated}\right.\\\left.\text{ to }\lceil-\log_{2}(\mathbb{P}_{\rs{q}}[q])\rceil+1\text{ bits. }\right)
\end{multline} If can be shown that for any distinct $q_{1},q_{2}\in\mathbb{A}$ with $\mathbb{P}_{\rvec{q}}[q_{1}],\mathbb{P}_{\rvec{q}}[q_{2}]>0$, $\mathrm{C}_{\rs{q}}^{\mathrm{U}}(q_{1})$ is not a prefix of $\mathrm{C}_{\rs{q}}^{\mathrm{U}}(q_{2})$ and vice-versa. The encoding $\mathrm{C}_{\rs{q}}^{\mathrm{U}}$ satisfies a prefix-property like that in Constraint \ref{pfxc:pf2}; namely the codewords are ``prefix-free" irrespective of the realization of $\rvec{\delta}$ \cite[Chapter 5.9]{elemIT}. This encoding scheme achieves a codeword length of $H(\rs{q})    \le \mathbb{E}_{\rs{q},\rs{\delta}}[\mathrm{C}^{\mathrm{U}}_{\rs{q}}(\rs{q})] \le H(\rs{q})+2$. It turns out that the upper bound on codeword length can be reduced is the encoder prepossesses $\rvec{q}$ to produce a random variable that is ``sorted" in order of decreasing probability mass. Assuming without loss of generality that $\mathbb{A}=\mathbb{N}$, let $s\colon\mathbb{A}\rightarrow\mathbb{A}$ be a bijection that re-indexes the support of $\rvec{q}$ such that  $\mathbb{P}_{\rs{q}}[s(0)]\ge \mathbb{P}_{\rs{q}}[s(1)]\ge \mathbb{P}_{\rs{q}}[s(2)]\ldots$. Such a bijection $s$ always exists, however it may be extremely difficult and/or computationally unreasonable to find \cite{verduVariableLength}. Let $\rvec{\overline{q}}= s(\rvec{q})$ and the function $F_{\rs{\overline{q}}}\colon\mathbb{A}\rightarrow [0,1)$ by ${F}_{\rs{\overline{q}}}(q) = \mathbb{P}_{\rs{\overline{q}}}[\rs{\overline{q}}<q]= \mathbb{P}_{\rs{q}}[\rs{\overline{q}}<q]$. Define the ``sorted, unconditional" SFE code $\mathrm{C}^{\mathrm{S}}:\mathbb{A}\rightarrow\{0,1\}^*$ by 
\begin{multline}\label{eq:sortedRecipe}
    \mathrm{C}^{\mathrm{S}}_{\rs{q}}(q) = \left(\text{the binary expansion of }  {F}_{\rs{\overline{q}}}(s(q))\text{ truncated}\right.\\\left.\text{ to }\lceil-\log_{2}(\mathbb{P}_{\rs{\overline{q}}}[s(q)])\rceil\text{ bits. }\right)
\end{multline} It can be shown that for distinct $q_{1}, q_{2}\in\mathbb{A}$, we have that  $a_{1}=\mathrm{C}_{\mathrm{q}}^{\mathrm{S}}(q_1)$ is not a prefix  $a_{2}=\mathrm{C}_{\mathrm{q}}^{\mathrm{S}}(q_2)$ and vice versa (cf. \cite[Problem 5.28]{elemIT}). We have that $H(\rvec{q})=H(\rvec{\overline{q}})$ and $ H(\rs{q})\le\mathbb{E}_{\rs{q}}[  \mathrm{C}^{\mathrm{S}}_{\rs{q}}(\rs{q})] \le H(\rs{q})+1$. We could also define a ``conditional sorted" codec $\mathrm{C}^{\mathrm{S}}_{\rs{q}|\rs{\delta}}$ which would allow the upper bound in (\ref{eq:unsortedsilength}) to be reduced by one bit. In general however, this would require the sorting function to depend on the realization of $\rvec{\delta}$. 
\subsubsection{Uniform quantizers with subtractive dither}\label{ssec:dithquant_key}
In this section, we introduce some key properties pertaining to element-wise uniform quantization with subtractive dither. These results are not new; many are generalizations of results from \cite{ditherQuant} described in detail in \cite{tanakaISIT}. Let $(\Delta\mathbb{Z})^{m}$ denote the set of $m$-tuples of integer multiples of $\Delta$, e.g., $r\in(\Delta\mathbb{Z})^{m}$ if, for some $r_0,\dots,r_{m-1}\in\mathbb{Z}$, $r = (r_0\Delta,r_1\Delta,\dots,r_{m-1}\Delta)$. Define an element-wise uniform quantizer with stepsize $\Delta>0$ as $Q_{\Delta}\colon\mathbb{R}^{m}\rightarrow(\Delta\mathbb{Z})^{m}$ via
\begin{align}\label{eq:quantizerDef}
    [Q_{\Delta}({x})]_{i} =  k\Delta, \text{ if }[x]_{i}\in[k\Delta-\Delta/2,k\Delta+\Delta/2),
\end{align} where $x\in \mathbb{R}^{m}$, $i\in\{0,\dots,m-1\}$. Let $\rs{z}$ be a random variable with range in $\mathbb{R}^{m}$. Let $\boldsymbol{\delta}=[\boldsymbol{\delta}_{0},\dots,\boldsymbol{\delta}_{m-1}]^{\mathrm{T}}$ be independent of $\rs{z}$ and such that the $[\boldsymbol{\delta}]_{i}$ are IID uniform on $[\frac{-{\Delta}}{2},\frac{\Delta}{2}]$. When $\rs{z}$ is quantized with an element-wise uniform quantizer with subtractive dither, the quantization is the random variable with range $(\Delta\mathbb{Z})^{m}$ defined by
\begin{align}
\rs{q} = Q_{\Delta}(\rs{z}+\boldsymbol{\delta}),
\end{align} the reconstruction is defined as $\rs{\tilde{q}} = \rs{q}-\boldsymbol{\delta}$, and the reconstruction error as $\rs{v}=\rs{\tilde{q}}-\rs{z}$. The following proposition summarizes some well-known, useful properties of dithered elementwise uniform quantizers. We use these properties to analyze the compression architecture of Fig. \ref{fig:overviewAchiev}.  
\begin{proposition}\label{prop:edqa}
Let $\rs{z}$, $\boldsymbol{\delta}$, $\rs{q}$,  $\rs{\tilde{q}}$, and $\rs{v}$ be as defined above.  Assume that $\mathbb{E}[\rs{z}]< \infty$, and that $\mathbb{E}[\rs{z}\rs{z}^{\mathrm{T}}]=Z$ where $Z \prec \infty$ . We have the following. 
\begin{enumerate}[label=(\roman*)]
    \item The $i^{\mathrm{th}}$ element of the reconstruction error $[\rs{v}]_{i}$ is uniformly distributed on the interval $[\frac{-{\Delta}}{2},\frac{\Delta}{2}]$. The $m$ elements of $\rs{v}$ are mutually independent, and $\rs{v}$ is independent of $\rs{z}$. \label{lemmclaim:work}
    \item We have $H(\rs{q})-H(\rs{q}|\rs{\delta})\le m$.\label{lemmclaim:dumpditherpenalty} 
    \item Let $\rvec{n}$ be a random vector whose elements are IID uniform random variables on $[-\frac{\Delta}{2},\frac{\Delta}{2}]$, and let $\rvec{n}\indep \mathbf{z}$. Let $N \in\mathbb{R}^{m\times m}$ be diagonal with $[N]_{i,i} = \Delta^2/12$. We have:
    \begin{IEEEeqnarray}{rCl}
        H(\rs{q}|\boldsymbol{\delta}) &=&  h(\rvec{z}+\rvec{n})-h(\rvec{n}) \label{eq:sublemmacapacity}\\ &=&h(\rvec{z}+\rvec{n}) +\frac{1}{2}\log_{2}\left(\frac{(\frac{2\pi e}{12})^{m}}{\det\left(2\pi e N\right)}\right)\label{eq:sublemmakl},
    \end{IEEEeqnarray} which implies that
    \begin{multline}
         H(\rs{q}|\boldsymbol{\delta}) \le\\ \frac{1}{2}\log_{2}\left(\frac{\det\left(Z+N\right)}{\det\left( N\right)}\right)+\frac{m}{2}\log_{2}\left(\frac{2\pi e}{12}\right).\label{eq:sublemmagauss}
    \end{multline}\label{lemmclaim:spaceFillingFixed}
\end{enumerate}
\end{proposition}
\begin{proof}
Claim \ref{lemmclaim:work} is a classic result. See \cite[Thm. 4.1.1]{zamir_book} for a general proof or \cite[Lemma 1]{tanakaISIT} for one specialized to this case. 
To see  \ref{lemmclaim:dumpditherpenalty}, note that $H(\rvec{q})-H(\rvec{q}|\rvec{\delta}) = I(\rvec{q};\rvec{\delta})$. Note also both $I((\rs{q},\rs{z});\rs{\delta}) =  I(\rs{q};\rs{\delta})+I(\rs{z};\rs{\delta}|\rs{q})$ and also $I((\rs{q},\rs{z});\rs{\delta}) = I(\rs{z};\rs{\delta})+I(\rs{q};\rs{\delta}|\rs{z})$, and thus as $\rvec{z}\indep\rvec{\delta}$, $I(\rvec{q};\rvec{\delta})\le I(\rs{q};\rs{\delta}|\rs{z})$. It is immediate that $I(\rs{q};\rs{\delta}|\rs{z}) = H(\rs{q}|\rs{z})$. Consider the scalar case $(m=1)$ and recognize that given $\mathbf{z}=z$, $\mathbf{q}$ can be determined to be in either the quantization ``bin" that contains $z$, or in one particular adjacent bin. Thus, for $m=1$, $H(\rs{q}|\rs{z})\le 1$. For a general $m$, the result follows as  $H(\rs{q}|\rs{z})\le \sum_{i=0}^{m-1}H([\rs{q}]_{i}|[\rs{z}]_{i})$. 

Equation (\ref{eq:sublemmacapacity}) in  \ref{lemmclaim:spaceFillingFixed} is well-established \cite[Theorem 5.2.1]{zamir_book}\cite[Lemma 1 (b)]{tanakaISIT}, (\ref{eq:sublemmakl}) follows from expanding $h(\rvec{n})$, and (\ref{eq:sublemmagauss}) follows as $\cov(\rvec{z}+\rvec{n}) = Z+N$ and Gaussian distributions have the maximum differential entropy among all distributions with the same covariance matrix. 
\end{proof}
We use Prop. \ref{prop:edqa} \ref{lemmclaim:spaceFillingFixed} and  \ref{lemmclaim:dumpditherpenalty} to develop  bounds on codeword length in the closed loop system of Fig. \ref{fig:overviewAchiev}. Prop., \ref{prop:edqa}\ref{lemmclaim:work} is likewise used to analyze the control performance. 
\subsection{Time-invariant near-achievability of the lower bound: Overview}\label{ssec:tiachiev}
In this section, we describe the internal variables in the closed loop system in Fig. \ref{fig:overviewAchiev}. Our description is sequential but necessarily recursive. Initially, we will abstract lossless source coding from the system; namely we will assume that at each time $t$ the encoder in Fig. \ref{fig:overviewAchiev} produces a discrete quantization, $\mathbf{q}_{t}$ which is conveyed exactly to the decoder. This leads naturally to an analysis of the system's incurred control cost. We then propose two strategies to losslessly encode the quantizations into prefix-free codewords $\{\mathbf{a}_{t}\}$ in a time-invariant manner. Finally, we state our main result, namely that these strategies can attain communication costs that nearly achieve the lower bound in Section \ref{sec:converse}.

Consider the system in Fig. \ref{fig:overviewAchiev}, and define $C$ and $V$ to be chosen optimally via the rate-distortion formulation in (\ref{eq:threestageRDF}). Since $C$ and $V$ are defined with respect to the minimizers of (\ref{eq:threestageRDF}) via (\ref{eq:cdef}), we can take $V = vI_{m}$ for some $v>0$ without loss of generality (defining $C$ so that $\hat{P}^{-1}-(A\hat{P}A^{\mathrm{T}}+W)^{-1}=C^{\mathrm{T}}V^{-1}C$, where $\hat{P}$ minimizes (\ref{eq:threestageRDF})). The encoder in Fig. \ref{fig:overviewAchiev} includes a elementwise uniform quantizer with sensitivity $\Delta$. The encoder and decoder share access to a common dither sequence of uniform random vectors, denoted $\{\rvec{\delta}_{t}\}$. The components of each $\rvec{\delta}_{t}$ vector are IID uniformly distributed on $[-\frac{\Delta}{2},\frac{\Delta}{2}]$ and the sequence $\{\rvec{\delta}_{t}\}$ is both IID over time and conforms to the conditional independence relationships implied by (\ref{eq:ditherFactorization}). With foresight, let the quantizer sensitivity and dither support be $\Delta = \sqrt{12v}$. 

In Fig \ref{fig:overviewAchiev}, both the encoder and the decoder operate identical time-invariant KFs. We denote the a priori and a posteriori estimates computed by these filters as $\rs{\overline{x}}_{t|t-1}$ and $\rs{\overline{x}}_{t|t}$, the corresponding estimator errors as $\rs{e}_{t}=\rs{x}_{t}-\rs{\overline{x}}_{t|t-1}$ and $\rs{e}_{t|t}=\rs{x}_{t}-\rs{\overline{x}}_{t|t}$, and the error covariance matrices $\overline{P}_{t|t-1}= \cov(\rs{e}_{t})$ and $\overline{P}_{t|t}=\cov(\rs{e}_{t|t})$. The initial a priori estimate is $\rs{\overline{x}}_{0|-1}=0$. The general intuition behind the architecture in Fig. \ref{fig:overviewAchiev} is that the state vector $\rvec{x}_{t}$, the estimates $\rvec{\overline{x}}_{t|t-1}$ and $\rvec{\overline{x}}_{t|t}$, and the control input $\rvec{u}_{t}$ are equivalent to those in the three-stage separation architecture of Fig. \ref{fig:threeStageTI} up to second order. We demonstrate this presently. 

Assume that at time $t$, the encoder and decoder's time-invariant KFs have identical a priori estimates  $\rs{\overline{x}}_{t|t-1}$ (this holds by design at $t=0$). We describe the system in Fig. \ref{fig:overviewAchiev} beginning from the encoder's input (the upper right of the figure) in a step-by-step fashion. 
\begin{enumerate}
    \item The encoder forms the linear measurement of the plant state, $C\rvec{x}_{t}$, and the associated Kalman innovation $C\rvec{e}_{t}$. 
    \item Assume that the dither sequence satisfies $\rvec{\delta}_{t}\indep (\rvec{a}^{t-1},\rvec{\delta}^{t-1},\rvec{u}^{t-1}, \rvec{w}^{t-1}, \rvec{e}^{t}, \rvec{x}^{t} )$. This is consistent with the assumptions in (\ref{eq:ditherFactorization}). The encoder then produces a \textit{dithered quantization} of the innovation, computing $\rvec{q}_{t} = Q_{\Delta}(C\rvec{e}_{t}+\rvec{\delta}_{t})$. 
    \item The encoder encodes $\rvec{q}_{t}$ (a discrete random variable) into the codeword $\rvec{a}_{t}\in\{0,1\}^{*}$ using a lossless source code. The codeword $\rvec{a}_{t}$ is conveyed to the decoder. As the coding is lossless, the decoder recovers $\rvec{q}_{t}$ exactly. 
    \item The decoder uses the recovered quantization and the common dither to compute the reconstruction $\rs{\tilde{q}}_{t} = \rs{q}_{t}-\rs{\delta}_{t}$. Let $\mathbf{v}_{t}=\rs{\tilde{q}}_{t}-C\rvec{e}_{t}$. From Proposition \ref{prop:edqa}, we have that $\rvec{v}_{t}$ is a vector with IID elements uniformly distributed on $[-\frac{\Delta}{2},\frac{\Delta}{2}]$ and that $\rvec{v}_{t}\indep \rvec{e}_{t}$. By assumption, the decoder-side KF's a priori estimate is also $\rs{\overline{x}}_{t|t-1}$. The decoder uses this to compute the centered measurement $    \rs{y}_{t} = \rs{\tilde{q}}_{t} + C\rs{\overline{x}}_{t|t-1}$, equivalently, 
\begin{IEEEeqnarray}{rCl}
    \rs{y}_{t} &=&  C\rs{x}_{t}+\rs{v}_{t}\label{eq:effectivelinearmodel}.
\end{IEEEeqnarray} Via (\ref{eq:effectivelinearmodel}), $\rvec{y}_{t}$ is a linear measurement of the plant state with additive uniform noise. Given that $\rvec{v}_{t}$ is a deterministic function of $\rvec{\delta}_{t}$ and $\rvec{e}_{t}$, we have $\rvec{\delta}_{t}\indep (\rvec{e}_{t},\rvec{x}_{t})$, and that $\rvec{v}_{t}\indep \rvec{e}_{t}$, it can be verified that $\rvec{v}_{t}\indep \rvec{x}_{t}$. The effective measurement matrix is $C$, and we have $\mathbb{E}[\mathbf{v}_{t}]=0_{m}$ and  $\cov(\rvec{v}_{t})=\frac{\Delta^2}{12}I=V$. 
\item Note that since the encoder has access to $\rvec{q}_{t}$, $\rvec{\delta}_{t}$, and $\rvec{\overline{x}}_{t|t-1}$, it can also compute the centered measurement $\rs{y}_{t}$. Both the decoder \textit{and} the encoder update their time-invariant KF estimate using $\rvec{y}_{t}$. Letting $J = \hat{P}_{+}{C}^{\mathrm{T}}({C}\hat{P}_{+}{C}^{\mathrm{T}}+{V})^{-1}$ as in Section \ref{sec:converse}, the encoder and decoder compute $\rs{\overline{x}}_{t|t} = \rs{\overline{x}}_{t|t-1}+J(\rs{y}_{t}-C\rs{\overline{x}}_{t|t-1})$.
\item  Let $K =-(B^{\mathrm{T}}SB+\Phi)^{-1}B^{\mathrm{T}}SA$ as in Sec. \ref{sec:converse}. 
The decoder forms the certainty-equivalent control input via $\rvec{u}_{t} = K\rvec{\overline{x}}_{t|t}$,
 which can also be computed at the encoder. The decoder feeds the control input into the plant, and both the encoder and decoder KFs compute prediction updates via $\rvec{\overline{x}}_{t+1|t} = A\rs{\overline{x}}_{t|t}+B\rvec{u}_{t}$.
Under this feedback arrangement, one can demonstrate that the sequence of reconstruction errors $\{\rs{v}_{t}\}$ are IID and that $\rvec{v}_{t}\indep \rvec{x}^{t}$ for all $t$. 
\end{enumerate}
Since the $\{\rvec{v}_{t}\}$ is a temporally white sequence with covariance $V$, and since  $\rvec{v}_{t}\indep \rvec{x}^{t}$, the linear measurement model in (\ref{eq:effectivelinearmodel}) is, to second order, identical to the one in the optimal three-stage test channel discussed in Section \ref{sec:converse}. The principal distinction is that in Fig. \ref{fig:overviewAchiev}, the measurement noise is uniform, rather than Gaussian (cf. (\ref{eq:effectivelinearmodel})). As the measurement models are the same to second order, the sequences of KF error covariance matrices, $\{\overline{P}_{t|t-1}\}$ and  $\{\overline{P}_{t|t}\}$ will satisfy the same recursions as the time-invariant KF in Section \ref{sec:converse}'s three-stage test channel. Thus, we have for $\hat{P}$ the minimizing $P$ from (\ref{eq:threestageRDF}) and $\hat{P}_{+}=A\hat{P}A^{\mathrm{T}}+W$, (cf. the discussion before (\ref{eq:cntrcostxpres})) $\lim_{t\rightarrow\infty} \overline{P}_{t|t-1} = \hat{P}_{+}\text{ and }\lim_{t\rightarrow\infty} \overline{P}_{t|t} = \hat{P}$. This leads to the following, via the equality preceding (\ref{eq:cntrcostxpres}).
\begin{proposition}\label{prop:controlcostandasymptote}
Consider the system of Fig. \ref{fig:overviewAchiev} as described above. So long as $\rvec{q}_{t}$ is recovered by the decoder at every $t$, the system attains $
        \underset{T\rightarrow \infty}{\lim}\frac{1}{T+1}\sum\nolimits_{t=0}^{T}\mathbb{E}[\lVert \rvec{x}_{t+1} \rVert_{{Q}}^{2} +\lVert \rvec{u}_{t} \rVert_{{R}}^{2}] \le \gamma.$
\end{proposition}
Regardless as to which lossless encoding scheme is used to encode the $\rvec{q}_{t}$ into the codewords $\rvec{a}_{t}$, Prop. \ref{prop:controlcostandasymptote} guarantees that the system in Fig. \ref{fig:overviewAchiev} achieves the desired constraint on LQG cost. 

In much of the prior work (cf. e.g. \cite{tanakaISIT}, \cite{kostinaTradeoff}), it was proposed to encode quantizations  $\{\rvec{q}_{t}\}$ using time-varying codebooks that were optimized, at every time $t$, to either the conditional PMF of $\rvec{q}_{t}$ given the dither realization $\rvec{\delta}_{t}$ or the unconditional PMF, i.e. producing codewords $\rvec{a}_{t}$ via e.g. $\rvec{a}_{t} = C^{\mathrm{U}}_{\rvec{q}_{t}|\rvec{\delta}_{t}}(\rvec{q}_{t}|\rvec{\delta}_{t})$ or $\rvec{a}_{t} = C^{\mathrm{S}}_{\rvec{q}_{t}}(\rvec{q}_{t})$. Time-asymptotic bounds on either ${\lim\sup}_{t\rightarrow\infty}\text{ } H(\rvec{q}_{t}|\rvec{\delta}_{t})$ or  ${\lim\sup}_{t\rightarrow\infty}\text{ } H(\rvec{q}_{t})$ were generally derived, and a \Cesaro mean argument then used to upper bound the time-average expected codeword length. As the $\{\rvec{q}_{i}\}$ are not identically distributed, these approaches are time-varying in that the mapping from quantizations $\rvec{q}_{t}$ (in the unconditioned case) or from quantizations and dither realizations $\rvec{\delta}_{t}$ (in the conditioned case) must generally vary at every $t$. Such ``perfect" adaptivity require great deal of computational overhead, and preclude arguments that suggest that the same bound on communication cost can be achieved with online, adaptive lossless coding schemes would seek to ``learn" the PMF of $\rvec{q}_{t}$ over time. This motivates an investigation of \textit{time-invariant} coding schemes. 

In this work, we propose to encode the $\{\rvec{q}_{t}\}$ in a ``time-invariant" manner. 
In one approach, we encode $\rvec{q}_{t}$ conditionally with an SFE code designed for a \textit{fixed} conditional distribution $\mathbb{P}_{\rvec{q}|\rvec{\delta}}$. In this case, the codewords $\rvec{a}_{t}$ are computed via $\rvec{a}_{t}=C^{\mathrm{U}}_{\rvec{q}|\rvec{\delta}}(\rvec{q}_{t}|\rvec{\delta})$. This approach will still satisfy Prefix Constraint \ref{pfxc:pf1}. While this approach is time-invariant in the sense that if $(\rvec{q}_{t},\rvec{\delta}_{t})=(x,y)$ and also $(\rvec{q}_{t+1},\rvec{\delta}_{t+1})=(x,y)$, then $\rvec{a}_{t}=\rvec{a}_{t+1}$,  using a ``conditional" codebook essentially requires that a different prefix-free codec (of the form (\ref{eq:unsorted})  or (\ref{eq:sortedRecipe})) be constructed for every potential realization of one of the $\rvec{\delta}_{t}$s (i.e. the conditional encoding uses the realization of the dither to select which codebook to use). For that reason, we also consider using a fixed time-invariant codebook of the form (\ref{eq:unsorted})  or (\ref{eq:sortedRecipe}) at all $t$. In other words, we ``unconditionally" encode $\{\rvec{q}_{t}\}$ with a fixed code of the form (\ref{eq:sortedRecipe}) designed using some fixed PMF $\mathbb{P}_{\rvec{q}}$, i.e. we assume that the codewords $\rvec{a}_{t}$ are given by  $\rvec{a}_{t} = C^{\mathrm{S}}_{\rvec{q}}(\rvec{q}_{t})$. Since a fixed prefix code is used at all $t$, the system will conform to Prefix Constraint \ref{pfxc:pf3}, which is the strongest, time-invariant constraint.

While generally speaking, the use of a fixed codebook would result in an increased codeword length, our main result is that for an unconditional (resp. conditional) codebook designed for a \textit{particular} fixed PMF (resp. conditional PMF) $\mathbb{P}_{\rvec{q}}$ (resp. $\mathbb{P}_{\rvec{q}|\rvec{\delta}})$, there is not an appreciable increase in communication cost. In particular, $\{\rvec{q}_{t},\rvec{\delta}_{t}\}$ is a Markov chain. We prove our main result by demonstrating that this chain has a limiting distribution, and that, in fact, encoding the $\rvec{q}_{t}$ with a lossless code adapted to the limiting PMF of $\rvec{q}_{t}$ (resp. conditional limiting PMF of $\rvec{q}_{t}$ given $\delta{q}_{t}$)  attains a communication cost close to the lower bound $\mathcal{R}(\gamma)$. The analysis also provides new ``almost sure" bounds on the time-average codeword length (as opposed to expected length). This result is summarized in the following, and is proven in  Section \ref{ssec:tiachievpf}. 
\begin{theorem}\label{thm:tiachiev}
\begin{enumerate}[label=(\roman*)]
    \item There exists a conditional PMF $\mathbb{P}_{\rs{q}|\rvec{\delta}}:(\Delta\mathbb{Z})^m\times [-\Delta/2,\Delta/2]^m \rightarrow [0,1]$ such that if $\mathrm{C}^{\mathrm{U}}_{\rs{q}|\rvec{\delta}}$ is as defined in (\ref{eq:unsortedRecipeSI}) with respect to $\mathbb{P}_{\rs{q}|\rvec{\delta}}$, and if the source codec in Fig. \ref{fig:overviewAchiev} encodes the quantization $\rs{q}_{t}$ with $\mathrm{C}^{\mathrm{U}}_{\rs{q}|\rvec{\delta}}$ given the dither $\rs{\delta}_{t}$ at every $t$ (i.e., $\rs{a}_{t} = \mathrm{C}^{\mathrm{U}}_{\rs{q}|\rvec{\delta}}(\rs{q}_{t}|\rs{\delta}_{t})$ for all $t$), the $\{\mathbf{a}_{t}\}$ satisfy Prefix Constraint \ref{pfxc:pf1} and their lengths will almost surely satisfy
    \begin{align}\label{eq:lastthmsi}
        \lim_{T\rightarrow\infty}\frac{1}{T+1}\sum_{i=0}^{T}\ell(\rs{a}_{t}) \le \mathcal{R}(\gamma)+\frac{m}{2}\log_{2}\left(\frac{2\pi e}{12}\right)+2, 
    \end{align} and furthermore
    \begin{multline}\label{eq:lastthmsi_kl}
        \lim_{T\rightarrow\infty}\frac{1}{T+1}\sum_{i=0}^{T}\mathbb{E}[\ell(\rs{a}_{t})] \le\\ \mathcal{R}(\gamma)+\frac{m}{2}\log_{2}\left(\frac{2\pi e}{12}\right)+2. 
    \end{multline}\label{lemmclaim:conditional}
    \item With $\mathbb{P}_{\rs{q}|\rvec{\delta}}$ defined as in \ref{lemmclaim:conditional}, define $\mathbb{P}_{\rs{q}}(q) = \frac{1}{\Delta^m}\int_{s\in[-\Delta/2,\Delta/2]^m}\mathbb{P}_{\rs{q}|\rvec{\delta}}(q|s)ds$.  Let $\mathrm{C}^{\mathrm{S}}_{\rs{q}}$ be the ``sorted" SFE code for $\rvec{q}_{\infty}$ as defined in (\ref{eq:sortedRecipe}) with respect to $\mathbb{P}_{\rs{q}}$. If the system in Fig. \ref{fig:overviewAchiev} uses $\mathrm{C}^{\mathrm{S}}_{\rs{q}}$ to encode the quantization $\rs{q}_{t}$ at every $t$  (i.e., $\rs{a}_{t} = \mathrm{C}^{\mathrm{S}}_{\rs{q}}(\rs{q}_{t})$ for all $t$), then the codewords $\{\rvec{a}_{t}\}$ will satisfy Prefix Constraint \ref{pfxc:pf3}, their lengths will almost surely satisfy
    \begin{multline}\label{eq:lastthm_nosi_ins}
        \lim_{T\rightarrow\infty}\frac{1}{T+1}\sum_{i=0}^{T}\ell(\rs{a}_{t}) \le\\ \mathcal{R}(\gamma)+m\left(1+\frac{1}{2}\log_{2}\left(\frac{2\pi e}{12}\right)\right)+1,
    \end{multline} and the time-average of expected codeword lengths satisfies 
       \begin{multline}\label{eq:lastthm_nosi_avg}
        \lim_{T\rightarrow\infty}\frac{1}{T+1}\sum_{i=0}^{T}\mathbb{E}[\ell(\rs{a}_{t}) ]\le\\ \mathcal{R}(\gamma)+m\left(1+\frac{1}{2}\log_{2}\left(\frac{2\pi e}{12}\right)\right)+1. 
    \end{multline} \label{lemmclaim:unconditional}
    \item Regardless of which lossless codec is used in Fig. \ref{fig:overviewAchiev}, in addition to the bound in Prop. \ref{prop:controlcostandasymptote}, the control cost almost surely satisfies $\underset{T\rightarrow\infty}{\lim\sup}\text{ }\frac{1}{T+1}\sum\nolimits_{t=0}^{T}\lVert \rvec{x}_{t+1} \rVert_{{Q}}^{2} +\lVert \rvec{u}_{t} \rVert_{{R}}^{2} < \gamma$.\label{lemmclaim:cci}
\end{enumerate}
\end{theorem}  In Theorem \ref{thm:tiachiev}, one can view $\mathbb{P}_{\rvec{q},\rvec{\delta}}$ as the limiting distribution of the Markov chain for $\{\rvec{q}_{t},\rvec{\delta}_{t}\}$. Theorem \ref{thm:tiachiev} provides two approaches to losslessly encode the quantizations $\rvec{q}_{t}$ that are notionally time-invariant. The approach in Theorem \ref{thm:tiachiev}\ref{lemmclaim:conditional} proposes to encode and decoder $\rvec{q}_{t}$ conditioned on the realization of the dither $\rvec{\delta}_{t}$, which is known at the decoder. In this approach, the prefix-free codebook used at each $t$ will generally change, however in contrast to the work in \cite{tanakaISIT}, the codec need not be adapted in both time \textit{and} with the dither realization. On the other hand, the approach in Theorem \ref{thm:tiachiev}\ref{lemmclaim:unconditional} is truly time-invariant. At every time $t$, $\rvec{q}_{t}$ is encoded with a fixed codebook, adapted to the limiting distribution of the $\{\rvec{q}_{t}\}$. This permits us to claim that this approach satisfies the ``time-invariant" Prefix Constraint \ref{pfxc:pf3}. Notably, Theorem \ref{thm:tiachiev} additionally provides an ``almost sure" bound on the realization of the time-average codeword length. In addition to bounds on the ``time average of expectations" communication cost defined in (\ref{eq:commcostxpres}), the bounds in (\ref{eq:lastthmsi}) and (\ref{eq:lastthm_nosi_ins}) imply that under the proposed encodings, the realizations of the long-term time average codeword lengths will almost surely satisfy the same upper bounds. The result for control performance in Theorem \ref{thm:tiachiev}\ref{lemmclaim:cci} is analogous. 
\subsection{Proof of Theorem \ref{thm:tiachiev}}\label{ssec:tiachievpf}
In this subsection, we establish a proof of Theorem \ref{thm:tiachiev}. We establish that the Markov chain $\{\rvec{q}_{t},\rvec{\delta}_{t}\}$ converges to some $(\rvec{q},\rvec{\delta})$. In particular, we demonstrate convergence is such that the time-average expected communication cost does not increase. These results follow from a long-term analysis of the stochastic process $\{\rs{e}_{t},\rs{\delta}_{t}\}$. Our analysis relies on well-established results from ergodic theory from \cite{ito_invariant} and \cite{mcmcReview}. 

Some properties of $\{\rs{e}_{t},\rs{\delta}_{t}\}$ will be especially useful. Let $L=AJ$ and $R = (A-LC)$. Recall that by definition  $\rs{v}_{t}= \rvec{\tilde{q}}_{t}-C\rvec{e}_{t}= Q_{\Delta}(C\rvec{e}_{t}+\rvec{\delta}_{t})-\rvec{\delta}_{t}-C\rvec{e}_{t}$. Define the function
$M\colon (x,y)\in\mathbb{D}^{m}\rightarrow\mathbb{R}^{m}$ via
\begin{align}\label{eq:MfuncDef}
    M(x,y) = Rx-L(Q_{\Delta}(Cx+y)-y-Cx).  
\end{align} 
Via (\ref{eq:ssmodel}) and the KF equations, it can be seen that $\{\rs{e}_{t}\}$ obeys the recursion
\begin{align}\label{eq:algebraicMassage}
    \rs{e}_{t} =  M(\rs{e}_{t-1},\rvec{\delta}_{t-1})+\rs{w}_{t-1},
\end{align} equivalently $\rs{e}_{t} = R\rs{e}_{t-1} -L\rs{v}_{t-1}+\rs{w}_{t-1}$. Since $\rs{\overline{x}}_{0|-1}=0_{m}$, $e_{0}\sim\mathcal{N}(0,X_{0})$, and as $({A},{W}^{\frac{1}{2}})$ is stabilizable and $({C},{A})$ is detectable, 
$R$ is stable with eigenvalues strictly inside the complex unit circle, i.e. $\meig(R)<1$ \cite{kailathBk}\cite{RDE_convergence}. Since $\rvec{w}_{t}\indep (\rvec{e}^{t},\rvec{\delta}^{t})$ and $\rvec{\delta}_{t+1}\indep (\rvec{e}^{t+1},\rvec{w}^{t})$, via (\ref{eq:algebraicMassage}), $\{\rvec{e}_{t},\rvec{\delta}_{t}\}$ is a time-homogeneous first order Markov chain on the state space $\mathbb{D}^m=\mathbb{R}^m\otimes[-\Delta/2,\Delta/2]^m$. The transition probabilities of the chain are described via a well-defined conditional PDF. Define the ``Gaussian PDF" function  $N(\dvec{r};\dvec{\mu},\Psi):\mathbb{R}^{m}\times \mathbb{R}^{m}\times  \mathbb{S}_{+}^{m}\rightarrow \mathbb{R}_{+}$ via 
 $N(\dvec{r};\dvec{\mu},\Psi)=\frac{1}{\sqrt{(2\pi)^{m}\det{\Psi}}}e^{-\frac{1}{2}(\dvec{r}-\dvec{\mu})^{\tp}\Psi^{-1}(\dvec{r}-\dvec{\mu})}$. To simplify notation, let $f_{t+1|t}=f_{\rs{e}_{t+1},\rs{\delta}_{t+1}|\rs{e}_{t},\rs{\delta}_{t}}$.  Via (\ref{eq:algebraicMassage}), the transition PDF 
$f_{t+1|t}\colon\mathbb{D}^{m}\times\mathbb{D}^{m}\rightarrow\mathbb{R}_{+}$ is
\begin{multline}\label{eq:transitionpdf}
  f_{t+1|t}({e}_{t+1},{\delta}_{t+1}|{e}_{t},{\delta}_{t}) =\\  \frac{1_{{\delta}_{t+1}\in [-\frac{\Delta}{2},\frac{\Delta}{2}]^m}}{\Delta^m}N(e_{t+1};M(e_{t},\delta_{t}),W),
\end{multline} where the indicator function in (\ref{eq:transitionpdf}) is ``always on" if $({e}_{t+1},{\delta}_{t+1})\in \mathbb{D}^m$, and is only included to emphasize that the support of each of the $\rvec{\delta}_{t}$ is the $m-$dimensional hypercube $[-\frac{\Delta}{2},\frac{\Delta}{2}]^m$. The transition PDF defines a well-defined regular conditional probability: for $\mathcal{K}\in\mathbb{B}(\mathbb{D}^{m})$, we have  $    \mathbb{P}_{\rs{e}_{t+1},\rs{\delta}_{t+1}|\rs{e}_{t},\rs{\delta}_{t}}[(\rs{e}_{t+1},\rs{\delta}_{t+1})\in\mathcal{K}|\rs{e}_{t},\rs{\delta}_{t}] \overset{\mathrm{a.s.}}{=}  \iint_{\mathbb{D}^m}1_{(x,y)\in\mathcal{K}} f_{t+1|t}(x,y|\rs{e}_{t},\rs{\delta}_{t})dxdy$. The Markov chain $\{\rvec{e}_{t},\rvec{\delta}_{t}\}$ has some useful properties that will be used to construct the encoding PMFs $\mathbb{P}_{\rvec{q}|\rvec{\delta}}$ and $\mathbb{P}_{\rvec{q}}$. Namely, the chain converges to an \textit{invariant measure} and has an ergodic property. These results are summarized in the following technical lemmas, proven in Appendix A. The proof of the first result uses the theory of weakly transient sets, namely \cite[Thm. 5]{ito_invariant}, to establish the existence of a potential limiting distribution. 
\begin{lemma}\label{lemm:invarexist}
The Markov chain on $\mathbb{D}^{m}$ defined by (\ref{eq:transitionpdf}) admits an invariant PDF; i.e., there exists a function $g_{\mathrm{inv}}\colon\mathbb{D}^{m}\rightarrow\mathbb{R}_+$ such that
\begin{multline}\label{eq:invarpdfdef}
    g_{\mathrm{inv}}(e_+,{\delta}_+) =\\ \iint_{(e,d)\in\mathbb{D}^{m}}f_{t+1|t}({e}_{+},{\delta}_{+}|{e},\delta)g_{\mathrm{inv}}(e,\delta)ded\delta
\end{multline} and $g_{\mathrm{inv}}(e,d)>0$ for all $(e,d)\in\mathbb{D}^{m}$. In other words, the Markov chain admits an invariant probability measure  $\mathbb{P}_{\mathrm{inv}}:\mathbb{B}(\mathbb{D}^{m})\rightarrow [0,1]$ defined by $\mathbb{P}_{\mathrm{inv}}[\mathcal{K}] = \iint_{(e,\delta)\in\mathcal{K}}g_{\mathrm{inv}}(e,\delta)ded\delta$
that is equivalent to the Lebesgue measure on $\mathbb{D}^{m}$ (i.e., $\mathbb{P}_{\mathrm{inv}}$ has a strictly positive PDF).
\end{lemma}
For intuition, note that if the initial conditions of a Markov chain are drawn from the invariant measure (e.g., $(\rs{e}_{0},\rs{\delta}_{0})\sim\mathbb{P}_{\mathrm{inv}}$) then for $i\ge1$ we will have $(\rs{e}_{i},\rs{\delta}_{i})\sim\mathbb{P}_{\mathrm{inv}}$. The next lemma states that if the initial conditions $(\rs{e}_{0},\rs{\delta}_{0})$ are continuous random variables, the $\mathbb{P}_{\rs{e}_{i},\rs{\delta}_{i}}$  converge to $\mathbb{P}_{\mathrm{inv}}$ and that an ergodic property holds. The analysis follows from \cite[Thm. 4]{mcmcReview}.
\begin{lemma}\label{lemm:ergo}
For $\lambda$ almost every initial condition, the $n$-step transition probabilities of the Markov chain defined by (\ref{eq:transitionpdf}) converge in total variation to the invariant measure, i.e., for $\lambda$ almost every $(e_{0},{\delta}_{0})$, $\lim_{t\rightarrow\infty} \sup_{\mathcal{K}\in\mathbb{B}(\mathbb{D}^{m})} |\mathbb{P}_{\mathrm{inv}}[ \mathcal{K}]-\mathbb{P}_{\rs{e}_{t},\rs{\delta}_{t}|\rs{e}_{0},\rs{\delta}_{0}}[\rvec{e}_{t},\rvec{\delta}_{t}\in\mathcal{K}|\rs{e}_{0}=e_{0},\rs{\delta}_{0}={\delta}_{0}]| = 0$. Furthermore, if $(\rs{e}_{0},\rs{\delta}_0)$ are continuous random variables then for any function $\theta\colon\mathbb{D}^{m}\rightarrow\mathbb{R}$ with $    \iint_{\mathbb{D}^{m}}|\theta(e,\delta)|g_{\mathrm{inv}}(e,\delta)ded\delta < \infty$, a ``law of large numbers" holds for $\{\rs{e}_{i},\rs{\delta}_{i}\}$ in the sense that $\lim_{T\rightarrow\infty}\frac{1}{T+1}\sum_{i=0}^{T}\theta(\rs{e}_{i},\rs{\delta}_{i}) \overset{\mathrm{a.s.}}{=} \mathbb{E}_{(\rs{e},\rs{\delta})\sim g_{\mathrm{inv}}}[\theta(\rs{e},\rs{\delta})]$.
\end{lemma}  Let $(\rvec{e},\rvec{\delta})\sim \mathbb{P}_{\mathrm{inv}}$, e.g. let $\mathbb{P}_{\rvec{e},\rvec{\delta}}[\rvec{e},\rvec{\delta}\in\mathcal{K}]=\mathbb{P}_{\mathrm{inf}}[\mathcal{K}]$ so that $(\rvec{e},\rvec{\delta})$ have the joint PDF $f_{\rvec{e},\rvec{\delta}}=g_{\mathrm{inv}}$. Since $(\rs{e}_{0},\rs{\delta}_{0})$ are continuous random variables on $\mathbb{D}^{m}$, an immediate consequence of Lemma \ref{lemm:ergo}'s convergence in total variation is that the sequence of $(\rs{e}_{t},\rs{\delta}_{t})$ converge in distribution to $(\rs{e},\rs{\delta})$. We now combine Lemmas \ref{lemm:invarexist} and \ref{lemm:ergo} to prove some useful facts about $\mathbb{P}_{\mathrm{inv}}$. 
\begin{corollary}\label{corr:invarprops}
Let $(\rs{e},\rs{\delta})\sim \mathbb{P}_{\mathrm{inv}}$. The marginal PDF of $\rs{e}$ is $f_{\rs{e}}(e) = \int_{[-\Delta/2,\Delta/2]^m}g_{\mathrm{inv}}(e,\delta)d\delta$. We have that $\rs{e}\indep \rs{\delta}$ and that $\rs{\delta}$ is a random vector whose elements are IID with $[\rvec{\delta}]_{i} \sim \text{Uniform}[-\Delta/2,\Delta/2]$. This implies that that the invariant PDF, $g_{\mathrm{inv}}$, factorizes via $g_{\mathrm{inv}}(e,d) = \frac{f_{\rvec{e}}(e)}{\Delta^{m}}$ for $(e,d)\in\mathbb{D}^{m}$. Furthermore, we have $\mathbb{E}[\rs{e}] = 0$ and $\mathbb{E}[\rvec{e}\rvec{e}^{\mathrm{T}}] = \hat{P}_{+}$.
\end{corollary}
\begin{proof}
If $\mathcal{A}$ is an open interval in $\mathbb{R}^m$ and $\mathcal{D}$ an open interval in $[-\Delta/2,\Delta/2]^{m}$ then $\mathcal{A}\times \mathcal{D} \in \mathbb{B}(\mathbb{D}^{m})$.
Using the definition of the invariant PDF (\ref{eq:invarpdfdef}) and the formula for $f_{t|t-1}$ from (\ref{eq:transitionpdf}), it can be shown that if $\mathcal{K}=\mathcal{A}\times\mathcal{D}$ then, $ \mathbb{P}_{\rs{e},\rs{\delta}}[\mathcal{K}] = \int_{\mathcal{A}}f_{\rvec{e}}(e)de\frac{\lambda(\mathcal{D})}{\Delta^{m}}$. By Dynkin's $\pi-\lambda$ theorem, this proves that $\rs{e}\indep \rs{\delta}$ (see e.g., \cite[Prop. 2.13]{zit}).  

Define $\rs{v}=\left(Q_{\Delta}(C\rs{e}+\rs{\delta})-(C\rs{e}+\rs{\delta})\right)$. By definition, $M(\rs{e},\rs{\delta})=R\rs{e}-L\rs{v}$. By the result just established, $\rs{\delta}\indep\rs{e}$ and the $[\rvec{e}]_{i}$ are IID uniformly distributed on $[-\Delta/2,\Delta/2]$. Thus, we can apply the properties of dithered quantizers from Prop. \ref{prop:edqa}.  Namely, by Prop. \ref{prop:edqa}\ref{lemmclaim:work} we have $\rs{v}\indep\rs{e}$ and that the components $[\rs{v}]_{i}$ are IID uniform random variables on $[-\Delta/2,\Delta/2]$.  It can be shown that 
\begin{IEEEeqnarray}{rCl}
    \mathbb{E}[\rvec{e}\rvec{e}^{\mathrm{T}}]&=& W+R\mathbb{E}[\rvec{e}\rvec{e}^{\tp}]R^{\mathrm{T}} +LVL^{\mathrm{T}}\label{eq:apedql}.
    \end{IEEEeqnarray}
The equality (\ref{eq:apedql}) follows from (\ref{eq:firstvarproof})-(\ref{eq:apedql2}) shown at the top of the following page.
\begin{figure*}[!t]
\normalsize
\begin{IEEEeqnarray}{rCl}
    \mathbb{E}[\rvec{e}\rvec{e}^{\mathrm{T}}] &=& \iint_{(e,d)\in\mathbb{D}^{m}}{e}{e}^{\mathrm{T}}g_{\mathrm{inv}}(e,\delta)ded\delta\label{eq:firstvarproof}
    \\ &=& \iint_{(e,d)\in\mathbb{D}^{m}}{e}{e}^{\mathrm{T}}\iint_{(s,t)\in\mathbb{D}^{m}}f_{t+1|t}({e},{\delta}|s,t)g_{\mathrm{inv}}(s,t)dsdtded\delta\label{eq:thedefinitionofinvar} \\
    &=& \iint_{(s,t)\in\mathbb{D}^{m}}\left(\iint_{(e,d)\in\mathbb{D}^{m}}{e}{e}^{\mathrm{T}}f_{t+1|t}({e},{\delta}|s,t)ded\delta \right)g_{\mathrm{inv}}(s,t)dsdt \label{eq:momentlemmatonelli}\\
    &=&\iint_{(s,t)\in\mathbb{D}^{m}} \left(W+M(s,t)M(s,t)^{\mathrm{T}}\right)g_{\mathrm{inv}}(s,t) dsdt\label{eq:normality}\\
    &=& W + \mathbb{E}_{(\rs{e},\rs{\delta})\sim \mathbb{P}_{\mathrm{inv}} }[(R\rvec{e}-L\rvec{v})(R\rvec{e}-L\rvec{v})^{\mathrm{T}}]\label{eq:vdef2}\\
    &=&W+R\mathbb{E}[\rvec{e}\rvec{e}^{\tp}]R^{\mathrm{T}} +LVL^{\mathrm{T}},\label{eq:apedql2}
\end{IEEEeqnarray}
\hrulefill
\vspace*{4pt}
\end{figure*}
 In particular, (\ref{eq:thedefinitionofinvar}) follows from the definition of the invariant PDF, (\ref{eq:momentlemmatonelli}) follows from the Fubini/Tonelli Theorem, (\ref{eq:normality}) follows from (\ref{eq:transitionpdf})  (i.e., since given $(\rs{e}_{t-1},\rs{\delta}_{t-1})$, $\rs{e}_{t}$ is normal with mean $M(\rs{e}_{t-1},\rs{\delta}_{t-1})$ and variance $W$), (\ref{eq:vdef2}) follows from (\ref{eq:MfuncDef}) and the definition of $\rs{v}$ above, and finally (\ref{eq:apedql2}) (equivalent to (\ref{eq:apedql})) follows from the aforementioned properties of $\rs{v}$ and the definition  $V=\frac{\Delta^2}{12}I_{m\times m}$. We recognize that the identity (\ref{eq:apedql}) is a Lyaponov equation in $\mathbb{E}[\rvec{e}\rvec{e}^{\mathrm{T}}]$. This equation has a unique PSD solution \cite[Prob. 4.9]{dullRobust}. It turns out that this unique solution to (\ref{eq:apedql}) is $\mathbb{E}[\rvec{e}\rvec{e}^{\mathrm{T}}]=\hat{P}$. To see this, note that by definition $\hat{P}_{+}$ satisfies the DARE
\begin{multline}\label{eq:darePrio}
      \hat{P}_{+}=\\ {A}\left( \hat{P}_{+}-\hat{P}_{+}{C}^{\mathrm{T}}({C}\hat{P}_{+}{C}^{{T}}+{V})^{-1}{C}\hat{P}_{+} \right){A}^{\mathrm{T}}+W.
\end{multline} Substituting the explicit formulas $R=A-LC$, $L=A\hat{P}_{+}{C}^{\mathrm{T}}({C}\hat{P}_{+}{C}^{\mathrm{T}}+{V})^{-1}$ and setting $\mathbb{E}[\rvec{e}\rvec{e}^{\mathrm{T}}] = \hat{P}_{+}$ in the right-hand side of  (\ref{eq:apedql})
exactly recovers the right-hand side of (\ref{eq:darePrio}). This proves the result. Since $\rs{e}\in\mathcal{L}^{2}$, we have $\rs{e}\in\mathcal{L}^{1}$. Given this, reductions analogous to (\ref{eq:firstvarproof}) through (\ref{eq:apedql}) demonstrate that $\mathbb{E}[\rs{e}] = R\mathbb{E}[\rs{e}]$. Since $\meig(R)<1$, it must be that $\mathbb{E}[\rs{e}] = 0_{m}$. 
\end{proof}
An immediate consequence of Lemma \ref{lemm:ergo} and the corollary is the ``almost sure" guarantee on the realization of the time-average control cost in Theorem \ref{thm:tiachiev}\ref{lemmclaim:cci}. By the lemma and corollary, we have that $\underset{T\rightarrow\infty}{\lim}\frac{1}{T+1}\sum\nolimits_{t=0}^{T}\lVert \rvec{x}_{t+1} \rVert_{{Q}}^{2} +\lVert \rvec{u}_{t} \rVert_{{R}}^{2} \overset{\mathrm{a.s.}}{=}  \mathrm{Tr}(\Theta \hat{P})+\mathrm{Tr}(WS)$. Since $\mathrm{Tr}(\Theta \hat{P})+\mathrm{Tr}(WS)<\gamma$, this proves Theorem \ref{thm:tiachiev}\ref{lemmclaim:cci}. 
With $(\rvec{e},\rvec{\delta})\sim\mathbb{P}_{\mathrm{inv}}$, let $\rvec{q}=Q_{\Delta}(C\rvec{e}+\rvec{\delta})$. The random variable $\rvec{q}$ is describes the quantizer output when its inputs are drawn from the invariant, limiting distribution. It can likewise be shown that the  $(\rvec{q}_{t},\rvec{\delta}_{t})$ converge in total variation to  $(\rvec{q}_{t},\rvec{\delta}_{t})$. Our general strategy is to design prefix-free codes for encoding the $\rvec{q}_{t}$ using the limiting conditional and unconditional PMFs $\mathbb{P}_{\rvec{q}|\rvec{\delta}}$ and $\mathbb{P}_{\rvec{q}}$. Both of these are well-defined;  namely for $r \in R^{m}$, let $\mathcal{B}_{\Delta}(r) = \{x\in\mathbb{R}^{m}:\lVert x-r\rVert_{\infty}\le \frac{\Delta}{2}\}$ denote a hypercube centered at $r$. For $z\in (\Delta\mathbb{Z})^{m}$, we have 
$\mathbb{P}_{\rvec{q}|\rvec{\delta}}[z|\rvec{\delta}=\delta] = \mathbb{P}_{\rvec{e}}[C\rvec{e}\in \mathcal{B}_{\Delta}(z-\delta)]$. Likewise, again for $z\in (\Delta\mathbb{Z})^{m}$, $\mathbb{P}_{\rvec{q}}[z] = \frac{1}{\Delta^m}\int_{\delta\in [-\frac{\Delta}{2},\frac{\Delta}{2}]^{m}}\mathbb{P}_{\rvec{e}}[C\rvec{e}\in \mathcal{B}_{\Delta}(z-\delta)] d\delta$. 

Assume first that the ``unconditional", ``sorted" encoding adapted to $\rvec{q}$ is used, i.e. at every $t$, $\rvec{a}_{t} = \mathrm{C}^{\mathrm{S}}_{\rvec{q}}(\rvec{q}_{t})$. By the definition of  $\mathrm{C}^{\mathrm{S}}_{\rvec{q}}(\rvec{q}_{t})$, the codeword length satisfies  $\ell(\rvec{a}_{t}) \le -\log_{2}\left(\mathbb{P}_{\rvec{q}}(\rvec{q}_{t})\right)+1$ .The ``law of large numbers" afforded by  Lemma \ref{lemm:ergo} gives  
\begin{IEEEeqnarray}{rCl}
    \underset{T\rightarrow \infty}{\lim\sup} \frac{1}{T+1}\sum_{i=0}^{T}\ell(\rs{a}_{t}) &\le&      \underset{T\rightarrow \infty}{\lim} \frac{\sum_{i=0}^{T}-\log_{2}\left(\mathbb{P}_{\rvec{q}}(\rvec{q}_{t})\right)}{T+1}+1\nonumber \\ &\overset{\mathrm{a.s.}}{\le}& H(\rvec{q})+1,\label{eq:applyentropydef}
\end{IEEEeqnarray} where (\ref{eq:applyentropydef}) follows since  $\mathbb{E}_{(\rvec{e},\rvec{\delta})\sim\mathbb{P}_{\mathrm{inv}}}[\mathbb{P}_{\rvec{q}}(\rvec{q})]=H(\rvec{q})$. At every $t$, the \textit{expected} codeword length satisfies $\mathbb{E}[\ell(\rvec{a}_{t})] \le \mathbb{E}_{\rvec{q}_{t}}\left[-\log_{2}\left(\mathbb{P}_{\rvec{q}_{t}}(\rvec{q}_{t})\right)+\log_{2}\left(\frac{\mathbb{P}_{\rvec{q}_{t}}(\rvec{q}_{t})}{\mathbb{P}_{\rvec{q}}(\rvec{q}_{t})} \right)+1\right]$; equivalently, we have 
\begin{IEEEeqnarray}{rCl}
    H(\rvec{q}_{t})+D_{\mathrm{KL}}(\rvec{q}_{t}||\rvec{q}) &\le& \mathbb{E}[\ell(\rvec{a}_{t})] \\&\le& H(\rvec{q}_{t})+D_{\mathrm{KL}}(\rvec{q}_{t}||\rvec{q})+1\label{eq:uncondwrongcode}.
\end{IEEEeqnarray} We will use these observations directly to establish Theorem \ref{thm:tiachiev}\ref{lemmclaim:unconditional}; namely we will use (\ref{eq:applyentropydef}) together with a bound on $H(\rvec{q})$ to establish (\ref{eq:lastthm_nosi_ins}). Likewise, to establish (\ref{eq:lastthm_nosi_avg}), we will bound ${\lim\sup}_{t\rightarrow\infty} H(\rvec{q}_{t})$ and prove that ${\lim\sup}_{t\rightarrow\infty}D_{\mathrm{KL}}(\rvec{q}_t||\rvec{q}) = 0$. Taking the \Cesaro mean then completes the argument. The analyses used to establish \ref{thm:tiachiev}\ref{lemmclaim:conditional} is completely analogous. If at every $t$, the system encodes $\rvec{q}_{t}$ given the realization of $\rvec{\delta}_{t}$ using an unsorted encoding adapted to $\mathbb{P}_{\rvec{q}|\rvec{\delta}}$, e.g. assume $\rvec{a}_{t}=\mathrm{C}^{\mathrm{U}}_{\rvec{q}|\rvec{\delta}}(\rvec{q}_{t}|\rvec{\delta}_{t})$. By the definition of $\mathrm{C}^{\mathrm{U}}_{\rvec{q}|\rvec{\delta}}(\rvec{q}_{t}|\rvec{\delta}_{t})$, the upper bound in (\ref{eq:applyentropydef}) is replaced by $\underset{T\rightarrow \infty}{\lim\sup} \frac{1}{T+1}\sum_{i=0}^{T}\ell(\rs{a}_{t}) \le H(\rvec{q}|\rvec{\delta})+2 $, and the bound in (\ref{eq:uncondwrongcode}) is replaced with \begin{multline}\label{eq:condwrongcode}
    H(\rvec{q}_{t}|\rvec{\delta}_{t})+D_{\mathrm{KL}}(\rvec{q}_{t}||\rvec{q}|\rvec{\delta}_{t}) \le\\ \mathbb{E}[\ell(\rvec{a}_{t})] \le H(\rvec{q}_{t}|\rvec{\delta}_{t})+D_{\mathrm{KL}}(\rvec{q}_{t}||\rvec{q}|\rvec{\delta}_{t})+2
\end{multline} where the conditional KL divergence is $D_{\mathrm{KL}}(\rvec{q}_{t}||\rvec{q}|\rvec{\delta}_{t}) = \mathbb{E}_{\rvec{q}_{t},\rvec{\delta}_{t}}[\log_{2}\left(\frac{\mathbb{P}_{\rvec{q}_{t}|\rvec{\delta}_{t}}(\rvec{q}_{t}|\rvec{\delta}_{t})}{\mathbb{P}_{\rvec{q}|\rvec{\delta}}(\rvec{q}_{t}|\rvec{\delta}_{t})}\right)]$. We bound $H(\rvec{q}|\rvec{\delta})$ to establish (\ref{eq:lastthmsi}), and we both bound ${\lim\sup}_{t\rightarrow\infty} H(\rvec{q}_{t}|\rvec{\delta}_{t})$ and prove that\\ ${\lim\sup}_{t\rightarrow\infty}D_{\mathrm{KL}}(\rvec{q}_t||\rvec{q}|\rvec{\delta}_{t}) = 0$ to establish (\ref{eq:lastthmsi_kl}). 
\begin{lemma}\label{lemm:entropyasymptote}
We have 
\begin{IEEEeqnarray}{rCl}
H(\rvec{q}|\rvec{\delta}) &\le& \mathcal{R}(\gamma)+\frac{m}{2}\log_2\left(\frac{2\pi e}{12}\right),\\
    H(\rvec{q}) &\le& \mathcal{R}(\gamma)+m+\frac{m}{2}\log_2\left(\frac{2\pi e}{12}\right),\\
\lim\sup_{t\rightarrow\infty} H(\rvec{q}_{t}|\rvec{\delta}_{t}) &\le& \mathcal{R}(\gamma)+\frac{m}{2}\log_2\left(\frac{2\pi e}{12}\right)\text{, and}\\ \lim\sup_{t\rightarrow\infty} H(\rvec{q}_{t}) &\le& \mathcal{R}(\gamma)+m+\frac{m}{2}\log_2\left(\frac{2\pi e}{12}\right).\end{IEEEeqnarray}
\end{lemma}
\begin{proof}
We first analyze $H(\rvec{q}|\rvec{\delta})$. Since by definition $\rvec{q}=Q_{\Delta}(C\rvec{e}+\rvec{\delta})$ and by Corollary \ref{corr:invarprops} $\rvec{e}\indep\rvec{\delta}$, we can apply Proposition \ref{prop:edqa}\ref{lemmclaim:spaceFillingFixed}. Setting $\rvec{z}=C\rvec{e}$ in the statement of Prop. \ref{prop:edqa}, noting that $\mathbb{E}[C\rvec{e}\rvec{e}^{\mathrm{T}}C^{\mathrm{T}}] = C\hat{P}C^{\mathrm{T}}$ by Corollary \ref{corr:invarprops}, and recalling that by definition $V=\frac{\Delta^2}{12}I_{m\times m}$, we have 
$ H(\rvec{q}|\rvec{\delta}) \le \frac{1}{2}\log_{2}\left(\det\left(C\hat{P}C^{\mathrm{T}}+V\right)\right)-\frac{1}{2} \log_2(\det(V))+\frac{m}{2}\log_{2}\left(\frac{2\pi e}{12}\right)$. Since by definition $\hat{P}^{-1} = \hat{P}_{+}^{-1} + C^{\mathrm{T}}V^{-1}C$ (see (\ref{eq:cdef})),  the matrix determinant lemma gives $ \det(C\hat{P}_{+}C^{\mathrm{T}}+V)=\det(\hat{P}_{+})\det(V) \det(\hat{P}^{-1})$. Since $\mathcal{R}(\gamma)=\log_{2}(\det(\hat{P}_{+}))+\log_{2}(\det(\hat{P}^{-1}))$ via (\ref{eq:commcostxpres}), we have $H(\rvec{q}|\rvec{\delta}) \le \mathcal{R}(\gamma) + \frac{m}{2}\log_{2}\left(\frac{2\pi e}{12}\right)$. By Prop. \ref{prop:edqa}\ref{lemmclaim:dumpditherpenalty},  $H(\rvec{q}) \le m + H(\rvec{q}|\rvec{\delta})$. The derivation of the bounds on $\lim\sup_{t\rightarrow\infty }H(\rvec{q}_{t}|\rvec{\delta}_{t})$ and $\lim\sup_{t\rightarrow\infty }H(\rvec{q}_{t})$ is completely analogous. Using Prop. \ref{prop:edqa}\ref{lemmclaim:spaceFillingFixed}, we have $H(\rvec{q}_{t}|\rvec{\delta}_{t}) \le  \log_{2}\left(\det\left(C\mathbb{E}\left[\rvec{e}_{t}\rvec{e}_{t}^{\mathrm{T}}\right]C^{\mathrm{T}}+V\right)\right)-\frac{1}{2} \log_2(\det(V))+\frac{m}{2}\log_{2}\left(\frac{2\pi e}{12}\right)$. Taking the limit of both sides, and recalling  that $\lim_{t\rightarrow \infty}\mathbb{E}\left[\rvec{e}_{t}\rvec{e}_{t}^{\mathrm{T}}\right] = \hat{P}$ gives the bound on $\lim\sup_{t\rightarrow\infty}H(\rvec{q}_{t}|\rvec{\delta}_{t})$. As  $H(\rvec{q}_{t})-H(\rvec{q}_{t}|\rvec{\delta}_{t}) \le m$, the bound on $\lim\sup_{t\rightarrow\infty}H(\rvec{q}_{t})$ follows.
\end{proof} From the preceding discussion (cf. (\ref{eq:applyentropydef})), Lemma \ref{lemm:entropyasymptote} proves the bounds on the realizations of time average codeword length in Theorem \ref{thm:tiachiev}'s (\ref{eq:lastthmsi}) and (\ref{eq:lastthm_nosi_ins}). To use a \Cesaro argument to establish  (\ref{eq:lastthmsi_kl}) and (\ref{eq:lastthm_nosi_avg}), we must demonstrate that the KL divergences $D_{\mathrm{KL}}(\left.\rvec{q}_{t}||\rvec{q}\right|\rvec{\delta}_{t}),D_{\mathrm{KL}}(\rvec{q}_{t}||\rvec{q})$ tend to $0$ as $t\rightarrow \infty$. This is the subject of the following lemma.  

\begin{lemma}
We have $\lim_{t\rightarrow\infty} D_{\mathrm{KL}}(\left.\rs{q}_{t}||\rs{q}\right|\rvec{\delta}_{t}) =0$ and $\lim_{t\rightarrow\infty} D_{\mathrm{KL}}(\rs{q}_{t}||\rs{q}) =0$. 
\end{lemma}
\begin{proof}
 It can be shown via Jensen's inequality that if $\rs{a},\rs{b}$ are random variables that are absolutely continuous with respect to Lebesgue measure such that $\rs{a}$ is absolutely continuous with respect to $\rs{b}$, then $D_{\mathrm{KL}}(Q_{\Delta}(\rs{a})||Q_{\Delta}(\rs{b}))\le D_{\mathrm{KL}}(\rs{a}||\rs{b})$. Thus, 
  we have $D_{\mathrm{KL}}(\rs{q}_{t}||\rs{q}) \le D_{\mathrm{KL}}(C\rvec{e}_{t}+\rvec{\delta}_{t}||C\rvec{e}+\rvec{\delta})$. Since $\rs{\delta}_{t}$ and $\rs{\delta}$ are identically distributed, $\rs{e}_{t}\indep\rs{\delta}_{t}$, and $\rs{e}\indep\rs{\delta}$,  the data processing inequality (DPI) for KL divergences (cf. \cite[Theorem 2.15]{polyanbk}) gives $  D_{\mathrm{KL}}(C\rs{e}_{t}+\rs{\delta}_{t}||C\rs{e}+\rs{\delta})\le D_{\mathrm{KL}}(\rs{e}_{t}||\rs{e})$.
The proof that $D_{\mathrm{KL}}(\left.\rs{q}_{t}||\rs{q}\right|\rvec{\delta}_{t})\le D_{\mathrm{KL}}(\rs{e}_{t}||\rs{e})$ is analogous. To begin, recognize that for each $\delta \in [-\frac{\Delta}{2},\frac{\Delta}{2}]^n$,  $D_{\mathrm{KL}}(\left.\rs{q}_{t}||\rs{q}\right|\rvec{\delta}_{t}=\delta)\le D_{\mathrm{KL}}(\left.C\rvec{e}_{t}+\delta|| C\rvec{e}+\delta \right|\rvec{\delta}_{t}=\delta)$ where both $\rvec{e}_{t}\indep \rvec{\delta}_{t}$ and $\rvec{e}\indep \rvec{\delta}_{t}$. Applying the DPI for every realization $\delta$ and using the fact that, by independence, $\mathbb{P}_{\rvec{e}_{t}|\rvec{\delta}_{t}} = \mathbb{P}_{\rvec{e}_{t}}$ and likewise  $\mathbb{P}_{\rvec{e}|\rvec{\delta}} = \mathbb{P}_{\rvec{e}}$ completes the argument. Thus, we can prove the lemma by demonstrating that $\lim_{t\rightarrow\infty} D_{\mathrm{KL}}(\rs{e}_{t}||\rs{e})=0$.

Let $\{\rs{\nu}_{t}\}$ denote an IID sequence of random variables uniformly distributed on $[-\Delta/2,\Delta/2]^m$, let $\{\rs{\omega}_{t}\}$ be IID with $\rs{\omega}_{t}\sim\mathcal{N}(0_{m},W)$, and let $\rs{\lambda}\sim\mathcal{N}(0_{m},X_{0})$. Assume $\{\rs{\omega}_{t}\}$, $\{\rs{\nu}_{t}\}$, and $\rs{\lambda}$ are mutually independent. Let ``$\overset{D}{=}$" denote ``equality in distribution", e.g., we write $\rs{a}\overset{D}{=}\rs{b}$ to imply $\rs{a}$ and $\rs{b}$ are identically distributed. From (\ref{eq:algebraicMassage}), we have $\rs{e}_{t} =R\rs{e}_{t-1} -L\rs{v}_{t-1}+\rs{w}_{t-1}$. Via Prop. \ref{prop:edqa}\ref{lemmclaim:work} and the factorization of system variables in (\ref{eq:ditherFactorization}), it can be verified that $\rs{w}_{t}\indep \rvec{e}^{t}, \rvec{v}^t,\rvec{w}^{t-1}$ and $\rs{v}_{t}\indep \rvec{e}^{t}, \rvec{v}^{t-1},\rvec{w}^{t}$.
Thus, by this recursive definition of $\{\rvec{e}_{t}\}$, $ \rs{e}_{t} \overset{D}{=} R^{t}\rs{\lambda}+\sum_{i=0}^{t-1}R^{i}(\rs{\omega}_{i}-L\rs{\nu}_{i}). $
 Likewise, by definition of $\rs{e}$, we have that both $\rs{e} \overset{D}{=} \lim_{t\rightarrow\infty}R^{t}\rs{\lambda}+\sum_{i=0}^{t-1}R^{i}(\rs{\omega}_{i}-L\rs{\nu}_{i})$ and $\rs{e} \overset{D}{=}\lim_{t\rightarrow\infty}\sum_{i=0}^{t-1}R^{i}(\rs{\omega}_{i}-L\rs{\nu}_{i})$, which follows since Lemma \ref{lemm:ergo}'s convergence in total variation implies weak convergence. 
Define the random variables $\rs{g}_{\le t} =  \sum_{i=0}^{t-1}R^{i}\rs{\omega}_{i}$,\\ $\rs{u}_{\le t} =  -\sum_{i=0}^{t-1}R^{i}L\rs{\nu}_{i}$, and $\rs{s}_{>t} =  \lim_{r\rightarrow \infty}\sum_{i=t}^{r}R^{i}(\rs{\omega}_{i}-L\rs{\nu}_{i})$   the limit is well defined by Kolmogorov's two-series theorem. By definition, $\rs{e}_{t} \overset{D}{=} R^{t}\rs{\lambda}+ \rs{g}_{\le t} +\rs{u}_{\le t}$ and $\rs{e} \overset{D}{=}  \rs{g}_{\le t}+ \rs{u}_{\le t} + \rs{s}_{> t}$.
Note that $\rs{g}_{\le t}\sim \mathcal{N}(0_{m},\sum_{i=0}^{t-1}R^{i}W(R^{i})^{\mathrm{T}})$.
We have
\begin{IEEEeqnarray}{rCl}
    D_{\mathrm{KL}}(\rs{e}_{t}|| \rs{e} ) &=&   D_{\mathrm{KL}}(R^{t}\rs{\lambda}+\rs{g}_{\le t} +\rs{u}_{\le t}||\rs{g}_{\le t} +\rs{u}_{\le t}+\rs{s}_{> t} )\nonumber \\ &\le& D_{\mathrm{KL}}(R^{t}\rs{\lambda}+\rs{g}_{\le t}||\rs{g}_{\le t} +\rs{s}_{> t} )\label{eq:usedidiv} \\ &\le& D_{\mathrm{KL}}(\left.R^{t}\rs{\lambda}+\rs{g}_{\le t} || \rs{g}_{\le t}+\rs{s}_{> t} \right| \rs{s}_{> t} ),\label{eq:condincdiv}  
\end{IEEEeqnarray} where (\ref{eq:usedidiv}) follows from the data processing inequality for KL divergence and (\ref{eq:condincdiv}) follows since conditioning increases KL divergence (see \cite[Theorem 2.14 (e)]{polyanbk}).

Given $\rs{s}_{> t}=s$,  (\ref{eq:condincdiv}) simplifies to a KL divergence between two $m-$dimensional multivariate Gaussians. Let $\Psi_t =\sum_{i=0}^{t-1}R^{i}W(R^{i})^{\mathrm{T}}$ and $\overline{\Psi}_{t}=\Psi_{t}+R^{t}X_{0}(R^{\tp})^{t}$. Since $\rs{\lambda}\indep\rs{g}_{\le t}$ by construction, $R^{t}\rs{\lambda}+\rs{g}_{\le t}\sim \mathcal{N}(0_{m},\overline{\Psi}_{t})$. Also by construction $(\rs{g}_{\le t}, \rs{\lambda}) \indep\rs{s}_{> t}$. Thus, we have 
\begin{multline}
D_{\mathrm{KL}}(\left.R^{t}\rs{\lambda}+\rs{g}_{\le t} || \rs{g}_{\le t}+\rs{s}_{> t} \right| \rs{s}_{> t} = s ) =\\ D_{\mathrm{KL}}(\mathcal{N}(0_{m},\overline{\Psi}_{t}) || \mathcal{N}(s,\Psi_{t}) ),\end{multline} and \begin{multline}
 D_{\mathrm{KL}}(\mathcal{N}(0_{m},\overline{\Psi}_{t}) || \mathcal{N}(s,\Psi_{t}) )=\\ \frac{1}{2}\log_{e}(\frac{\det\Psi_{t}}{\det\overline{\Psi}_{t}})+\frac{\text{Tr}(\Psi_{t}^{-1}\overline{\Psi}_{t})}{2}+\frac{s^{\mathrm{T}}\Psi_{t}^{-1}s}{2}-\frac{m}{2},\label{eq:sayinnats}
\end{multline} where the divergence in (\ref{eq:sayinnats}) is in nats. Let $d_{t}= \log_{e}(\frac{\det\Psi_{t}}{\det\overline{\Psi}_{t}})+ \text{Tr}(\Psi_{t}^{-1}\overline{\Psi}_{t})-m$.
Taking the expectation over realizations $s$, we have 
\begin{multline}
    D_{\mathrm{KL}}(\left.R^{t}\rs{\lambda}+\rs{g}_{\le t} || \rs{g}_{\le t}+\rs{s}_{> t} \right| \rs{s}_{> t} ) =\\ \frac{1}{2}(d_{t}+\frac{\mathbb{E}_{\rs{s}_{>t}}[\rs{s}_{>t}^{\mathrm{T}}\Psi_{t}^{-1}\rs{s}_{>t}]}{2}).\label{eq:expectedkl}
\end{multline} It is immediate that $\rs{s}_{>t}\in\mathcal{L}^2$, so (\ref{eq:expectedkl}) is always finite. We analyze each of the terms in (\ref{eq:expectedkl}) in turn. Since $\overline{\Psi}_{t} \succeq \Psi_{t} \succeq W \succ 0_{m\times m }$, we have that $\det\Psi_{t},\det\overline{\Psi}_{t}>0$. Since $R$ is globally asymptotically stable (with $\meig(R)<1$), we have well defined, equal limits $\lim_{t\rightarrow\infty} \Psi_{t} = \Psi_{\infty}$ and $\lim_{t\rightarrow\infty} \overline{\Psi}_{t} = \Psi_{\infty}$ (see Proposition A.4 in Appendix A). Thus, $\lim_{t\rightarrow\infty }\log_{e}(\frac{\det\Psi_{t}}{\det\overline{\Psi}_{t}}) = 0$ and $\lim_{t\rightarrow\infty }\text{Tr}(\Psi_{t}^{-1}\overline{\Psi}_{t}) = m$,  implying $\lim_{t\rightarrow\infty }d_{t}=0$.

We now establish that $\lim_{t\rightarrow\infty}\mathbb{E}_{\rs{s}_{>t}}[\rs{s}_{>t}^{\mathrm{T}}\Psi_{t}^{-1}\rs{s}_{>t}] = 0$. Let $\rvec{p}_{t:r} = \sum_{i=t}^{r}R^{i}(\rvec{\omega}_{i}-L\rvec{\nu}_{i})$. For any $t$, by definition $\lim_{r\rightarrow\infty }\rvec{p}_{t:r}\rvec{p}_{t:r}^{\mathrm{T}} = \rvec{s}_{>t}\rvec{s}_{>t}^{\mathrm{T}}$, where we again note that the limit is well defined by Kolmogorov's two-series theorem. Then, we then have for any $t$ 
\begin{IEEEeqnarray}{rCl}
    \mathbb{E}[\rs{s}_{>t}^{\mathrm{T}}\Psi_{t}^{-1}\rs{s}_{>t}] &=&  \mathbb{E}\left[\text{Tr}\left(\Psi_{t}^{-1}\lim_{r\rightarrow\infty }\rvec{p}_{t:r}\rvec{p}_{t:r}^{\mathrm{T}}\right)\right] \\ &\le& \text{Tr}\left(\Psi_{t}^{-1} \underset{r\rightarrow \infty}{\lim\inf}\text{ }\mathbb{E}[\rvec{p}_{t:r}\rvec{p}_{t:r}^{\mathrm{T}}] \right)\label{eq:usefatou},
\end{IEEEeqnarray} where (\ref{eq:usefatou}) follows from Fatou's lemma and the linearity of the trace/expectation. Let $\Gamma =\lim_{j\rightarrow\infty}\sum_{i=0}^{j}R^{i}\left(W+LVL^{\mathrm{T}}\right)(R^{\tp})^{i}$,
where the limit is well defined since $R$ has $\meig(R)<1$. It is easy to see directly that $\underset{r\rightarrow \infty}{\lim}\text{ }\mathbb{E}[\rvec{p}_{t:r}\rvec{p}_{t:r}^{\mathrm{T}}] = R^{t}\Gamma (R^{\tp})^{t}$. Consequently, from (\ref{eq:usefatou}), we have
\begin{align}\label{eq:takethelimitofme}
    \mathbb{E}[\rs{s}_{>t}^{\mathrm{T}}\Psi_{t}^{-1}\rs{s}_{>t}] \le \text{Tr}\left(\Psi_{t}^{-1} R^{t}\Gamma (R^{\tp})^{t} \right). 
\end{align} It is immediate that $\Psi^{-1}_{t} \preceq W^{-1}$. Since $\meig(R)<1$ taking the limit of both sides of (\ref{eq:takethelimitofme}) as $t\rightarrow \infty$ gives $\lim_{t\rightarrow\infty} \mathbb{E}[\rs{s}_{>t}^{\mathrm{T}}\Psi_{t}^{-1}\rs{s}_{>t}]=0$. Since $\lim_{t\rightarrow \infty}d_{t}=0$, taking the limit of both sides of (\ref{eq:expectedkl}) as $t\rightarrow \infty$ gives that $\lim_{t\rightarrow \infty}D_{\mathrm{KL}}(\left.R^{t}\rs{\lambda}+\rs{g}_{\le t} || \rs{g}_{\le t}+\rs{s}_{> t} \right| \rs{s}_{> t} ) =0$. Since  $D(\rvec{e}_{t}||\rvec{e})\le D_{\mathrm{KL}}(\left.R^{t}\rs{\lambda}+\rs{g}_{\le t} || \rs{g}_{\le t}+\rs{s}_{> t} \right| \rs{s}_{> t})$ this proves the lemma.
\end{proof} Since $D_{\mathrm{KL}}(\rs{q}_{t} || \rs{q} )$ and $D_{\mathrm{KL}}(\left.\rs{q}_{t} || \rs{q} \right|\rvec{\delta}_{t})$ both tend to to $0$ as $t\rightarrow \infty$, taking the time averages of (\ref{eq:uncondwrongcode}) and (\ref{eq:condwrongcode}) and applying \Cesaro means gives (\ref{eq:lastthm_nosi_avg}) and (\ref{eq:lastthmsi_kl}) respectively. 
\section{Conclusion}\label{sec:conclcusion}
In this work we demonstrated that dithered quantization can enable a time-invariant encoding architecture to achieve near minimum bitrate prefix-free feedback in LQG control systems. There are several interesting opportunities for future work. An extension of our time-invariant achievability argument to nonsingular codes is essentially immediate. In both the conditional and unconditional ``time-invariant" approaches presented in this work, the difference between the upper and lower bounds on time average bitrate is linear in plant dimension (e.g. for the fully time-invariant scheme of Theorem \ref{thm:tiachiev}\ref{lemmclaim:unconditional}, the upper bound in (\ref{eq:lastthm_nosi_avg}) is about $1+1.26m$ bits above the lower bound $\mathcal{R}(\gamma)$). In the time-varying (but dither free) scheme  in \cite{kostinaTradeoff}, the gap between upper and lower bounds is $\mathcal{O}(\log(m))$). This follows from \cite{kostinaTradeoff}'s use of more sophisticated lattice quantizers \cite{kostinaTradeoff}. We believe that using (dithered) lattice quantizers in place of uniform quantizers in the present setup could reduce the scaling of our upper bounds with plant dimension.  Another opportunity is to explore the ergodic properties of the quantizer output in the achievability approach proposed in \cite{kostinaTradeoff}; this could lead to a dither-free time-invariant achievability result. 

Another opportunity is to expand this work to a more general class of MIMO plants. An extension to partially observed plants (where the encoder has access only to a noisy measurement of the plant) requires a modified converse (lower-bound) analysis. An reasonable staring point for this line of research is the rate distortion formulation in \cite[Section VII]{SDP_DI}. It is notable that in several areas, our proofs rely on the fact the process noise covariance is full rank (e.g. $W\succ 0_{m\times m}$); in particular this assumption is used liberally in establishing Lemmas \ref{lemm:invarexist} and \ref{lemm:ergo}. A starting point for relaxing this assumption is the rate-distortion formulation of \cite[Thm. 1]{SDP_DI}, which could be used to design an optimal test channel akin to that of Section \ref{sec:converse}. It would also be useful to formulate a non-time-asymptotic analysis of the convergence of communication and control costs in our proposed approach.

Finally, it would also be interesting to examine adaptive zero-delay source coding codecs in our present context; it seems likely that the properties of the invariant measure established in Section \ref{ssec:tiachievpf} may be useful in analyzing the asymptotic redundancy of such approaches.


\clearpage


\ifCLASSOPTIONcaptionsoff
  \newpage
\fi

\bibliographystyle{IEEEtran}
\bibliography{IEEEabrv,references}

\begin{IEEEbiography}[{\includegraphics[width=1in,height=1.25in,clip,keepaspectratio]{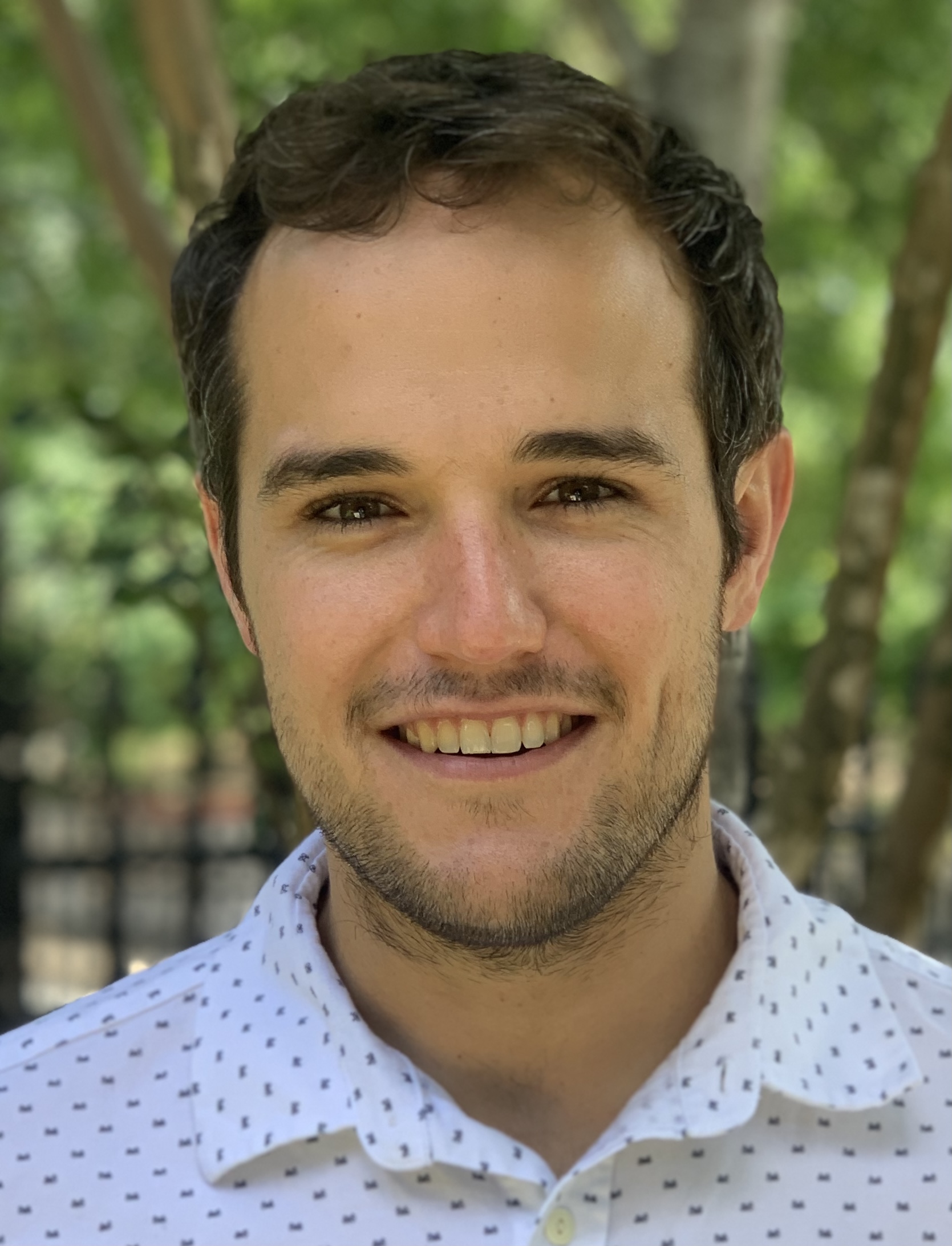}}]%
{Travis C. Cuvelier} (S'13)  received the B.S. and M.Eng. degrees in Electrical and Computer Engineering from Cornell University, Ithaca, NY, in 2015 and 2016. Since 2016, he has been pursuing a Ph.D. in the  Department of Electrical and Computer Engineering at the University of Texas at Austin. He previously held internships at LGS Innovations and The MITRE Corporation. At UT, he is affiliated with the Wireless Networking and Communications Group, the Oden Institute for Computational Engineering and Sciences, and the Applied Research Laboratories. His research interests include broad areas of signal processing and information theory with applications to network control systems and wireless communications. 
\end{IEEEbiography}

\begin{IEEEbiography}[{\includegraphics[width=1in,height=1.25in,clip,keepaspectratio]{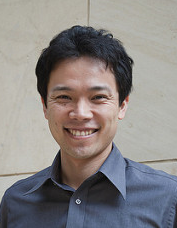}}]%
 {Takashi Tanaka} (SM'22) received the B.S. degree in aerospace engineering from the University of Tokyo, Tokyo, Japan, in 2006, and the M.S. and Ph.D. degrees in aerospace engineering(automatic control) from the University of Illinois at Urbana–Champaign, Champaign, IL, USA, in 2009 and 2012, respectively.
From 2012 to 2015, he was a Postdoctoral Associate with the Laboratory for Information and Decision Systems, Massachusetts Institute of Technology, Cambridge, MA, USA. From 2015 to 2017, he was a Postdoctoral Researcher with the KTH Royal Institute of Technology, Stockholm, Sweden. Since 2017, he has been an Assistant Professor with the Department of Aerospace Engineering and Engineering Mechanics, University of Texas at Austin, Austin, TX, USA.
\end{IEEEbiography}

\begin{IEEEbiography}[{\includegraphics[width=1in,height=1.25in,clip,keepaspectratio]{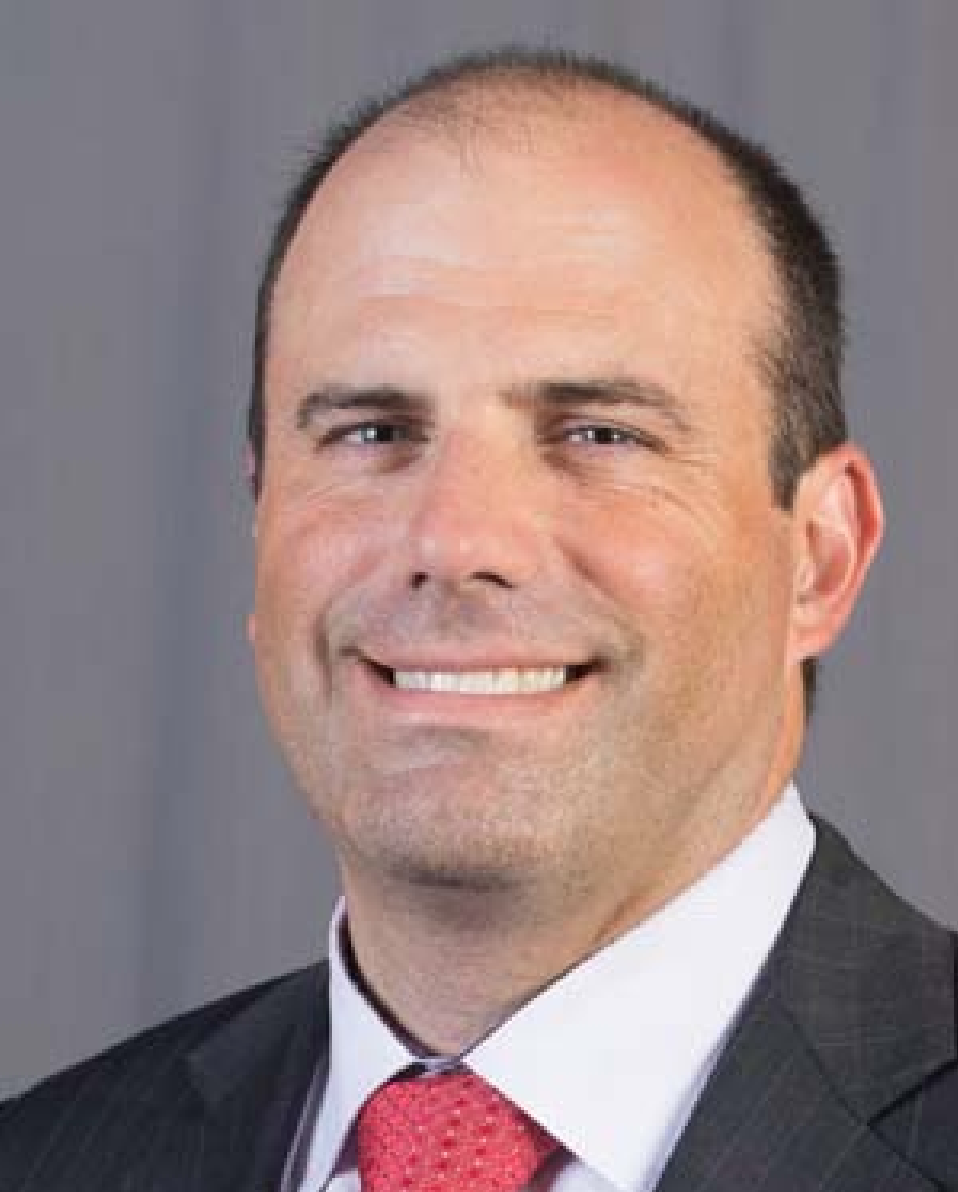}}]%
{Robert W. Heath Jr.}(F'11) received the B.S. and M.S. degrees from the University of Virginia, Charlottesville, VA, in 1996 and 1997, respectively, and the Ph.D. degree from Stanford University, Stanford, CA, in 2002, all in electrical engineering. From 1998 to 2001, he was a Senior Member of the Technical Staff then a Senior Consultant with Iospan Wireless, Inc., San Jose, CA, where he worked on the design and implementation of the physical and link layers of the first commercial MIMO-OFDM communication system. From 2002-2020, he was with The University of Texas at Austin, most recently as Cockrell Family Regents Chair in Engineering and Director of UT SAVES. He is presently a Distinguished Professor with North Carolina State University. He is also President and CEO of MIMO Wireless Inc. He authored Introduction to Wireless Digital Communication (Prentice Hall, 2017) and Digital Wireless Communication: Physical Layer Exploration Lab Using the NI USRP (National Technology and Science Press, 2012), and coauthored Millimeter Wave Wireless Communications (Prentice Hall, 2014) and Foundations of MIMO Communication (Cambridge University Press, 2018). He is currently Editor-in-Chief of IEEE Signal Processing Magazine and is a member-at-large of the IEEE Communications Society Board of Governors.
Dr. Heath has been a coauthor of a number award winning conference and journal papers including recently the 2016 IEEE Communications Society Fred W. Ellersick Prize, the 2016 IEEE Communications and Information Theory Societies Joint Paper Award, the 2017 Marconi Prize Paper Award, and the 2019 IEEE Communications Society Stephen O. Rice Prize. He was the recipient of the 2017 EURASIP Technical Achievement award and the 2019 IEEE Kiyo Tomiyasu Award. He was a Distinguished Lecturer and member of the Board of Governors in the IEEE Signal Processing Society. In 2017, he was selected as a Fellow of the National Academy of Inventors. He is also a licensed Amateur Radio Operator, a Private Pilot, a registered Professional Engineer in Texas.
\end{IEEEbiography}
\clearpage
\appendices

\section{Proofs of Technical Lemmas}\label{apx:ergopfs}
To prove Lemmas \ref{lemm:invarexist} and \ref{lemm:ergo}, it is useful to denote the $n-$step transition PDFs \begin{align}
     f_{t+n|t}(e,\delta|e_{p},\delta_{p})=f_{\rvec{e}_{t+n},\rvec{\delta}_{t+n}|\rvec{e}_{t},\rvec{\delta}_{t}}(e,\delta|e_{p},\delta_{p}).
\end{align} Applying the standard Chapman-Kolmogorov equations to (\ref{eq:transitionpdf}), it can be seen that the $n-$step transition PDFs satisfy, for $(e,\delta),(e_{p},\delta_{p}) \in \mathbb{D}^{m}\times\mathbb{D}^{m}$, 
\begin{IEEEeqnarray}{rCl}
    f_{t+n|t}(e,\delta|e_{p},\delta_{p}) &=& f_{\rvec{e}_{t+n}|\rvec{e}_{t},\rvec{\delta}_{t}}(e|e_{p},\delta_{p})f_{\rvec{\delta}_{t+n}}(\delta)\\ &=& f_{\rvec{e}_{t+n}|\rvec{e}_{t},\rvec{\delta}_{t}}(e|e_{p},\delta_{p})\frac{1_{\delta \in [-\frac{\Delta}{2},\frac{\Delta}{2}]^m}}{\Delta^m}.\label{eq:nStepTransition}
\end{IEEEeqnarray} %
\subsection{Proof of Lemma \ref{lemm:invarexist}}
We prove the existence of the invariant PDF using results from \cite{ito_invariant}. Formally speaking, we use the results of \cite{ito_invariant} to verify that Markov chain described by (\ref{eq:transitionpdf})
has an invariant measure that is equivalent to the Lebesgue measure (i.e., it has a PDF that is strictly positive).  Generally speaking, when restating definitions and theorems from \cite{ito_invariant}, we will not do so in full generality but rather adapt them to the present setting.  We begin with a definition. 
\begin{definition}[{\cite[Definition 5]{ito_invariant}}]\label{def:weakTrans}
A set $F\in\mathbb{B}(\mathbb{D}^m)$ is called \textit{weakly transient} with respect to the Markov kernel $\mathbb{ P}_{\rvec{e}_{t+1},\rvec{\delta}_{t+1}|\rvec{e}_{t},\rvec{\delta}_{t}}$ if there exists a sequence of positive integers $n_{1}<n_{2}<\dots$ such that 
\begin{align}
    \sum_{i=1}^{\infty}\mathbb{P}_{\rs{e}_{n_{i}},\rs{\delta}_{n_i}|\rs{e}_{0},\rs{\delta}_{0}}[F|\rs{e}_{0}={e}_{0},\rs{\delta}_{0}={\delta}_{0}]<\infty
\end{align} holds for $\lambda$ almost-every $({e}_{0},{\delta}_{0})$. 
\end{definition}  A key result from  \cite{ito_invariant} is the following.
\begin{theorem}[{\cite[Theorem 5]{ito_invariant}}]\label{thm:existanceCrit}
There exists an invariant PDF $g_{\mathrm{inv}}$ satisfying (\ref{eq:invarpdfdef}) and\\ $g_{\mathrm{inv}}(a,b)>0$ for all $(a,b)\in\mathbb{D}^{m}$ if and only if every weakly transient set $F$ has $\lambda(F) = 0$.  
\end{theorem} We prove that the invariant PDF exists by demonstrating that under the Markov model (\ref{eq:transitionpdf}), any weakly transient set must have Lebesgue measure $0$. Recall from the discussion in Section \ref{ssec:dithquant_key} that the reconstruction at time $t$ is given by
\begin{align}
    \rs{\tilde{q}}_t = {Q}_{\Delta}(C\rs{e}_{t}+\rs{\delta}_{t})-\rs{\delta}_{t}
\end{align} and the reconstruction error is then $\rs{v}_{t}=\rs{\tilde{q}}_t-C\rs{e}_{t}$. Recall $R= (A-LC)$, and that $\meig(R)<1$. The first lemma derives the functional form of a particular conditional PDF that will arise in future calculations.  
\begin{lemma}\label{lemm:normalPDF}
Let 
\begin{align}\label{eq:mundef}
     \mu_n(e_0,{\delta}_0,v_1^{n-1})= R^{n-1}M(e_0,{\delta}_0)-\sum\limits_{i=0}^{n-2}R^iLv_{n-1-i}
\end{align} and 
\begin{align}\label{eq:sigmandef}
    \Sigma_{n} = \sum\limits_{i=0}^{n-1}R^iW(R^{\tp})^{i}.
\end{align}
For all $n\ge 1$ we have
\begin{multline}\label{eq:normalStepFormula}
    f_{\rvec{e}_{n}|\rvec{e}_{0},\rvec{\delta}_{0},\rvec{v}_1^{n-1}}({e}_{n}|{e}_{0},{\delta}_{0},{v}_1^{n-1}) =\\ N(e_{n};  \mu_n(e_0,{\delta}_0,v_1^{n-1}),\Sigma_{n})
\end{multline} where $M$ was defined in (\ref{eq:MfuncDef}), and by convention $\rs{v}_{1}^{z}=\emptyset$ if $z \le 0$, $\sigma_1^2 = W$, and  $\mu_1= M(e_0,{\delta}_0)$. 
\end{lemma}
\begin{proof}
The proof follows from induction on $n$. The base case for $n=1$ is readily established from (\ref{eq:transitionpdf}) after marginalizing over $\rs{\delta}_{1}$. Assume the formula (\ref{eq:normalStepFormula}) holds for $n=k-1$. We demonstrate that it must hold for $n=k$. We have via Bayes' Theorem that
\begin{multline}\label{eq:bayesthm}
    f_{\rs{e}_{k}|\rs{e}_0,\rs{\delta}_0,\rs{v}_{1}^{k-1}}({e}_{k}|{e}_0,{\delta}_0,{v}_{1}^{k-1}) =\\ \frac{f_{\rs{e}_{k},\rs{v}_{k-1}|\rs{e}_0,\rs{\delta}_0,\rs{v}_{1}^{k-2}}({e}_{k},v_{k-1}|{e}_0,{\delta}_0,{v}_{1}^{k-2})}{f_{\rs{v}_{k-1}|\rs{e}_0,\rs{\delta}_0,\rs{v}_{1}^{k-2}}(v_{k-1}|{e}_0,{\delta}_0,{v}_{1}^{k-2})}
\end{multline}
Since  $\rs{v}_{k-1}$ is a measurable function of $\rvec{e}_{k-1}$ and $\rvec{\delta}_{k-1}$ and $\rvec{\delta}_{k-1}\indep (\rvec{e}_{k-1}, \rs{v}_{1}^{k-2}, \rs{e}_{0}, \rs{\delta}_0 )$, it can be seen that $\rs{v}_{k-1}\indep(\rs{v}_{1}^{k-2}, \rs{e}_{0}, \rs{\delta}_0 )$ \textit{given} $\rs{e}_{k-1}$. By the properties of dithered quantizers in Prop. \ref{prop:edqa}\ref{lemmclaim:work},  we have that $\rs{v}_{k-1}$ is (pairwise) independent of $\rs{e}_{k-1}$. Together, these imply that $\rvec{v}_{k-1}\indep (\rvec{e}_{k-1},\rvec{e}_{0},\rvec{\delta}_{0},\rvec{v}^{k-2}_{1})$. Thus,  suppressing the implicit dependence on realizations, we can derive
\begin{multline}
        f_{\rs{e}_{k},\rs{v}_{k-1}|\rs{e}_0,\rs{\delta}_0,\rs{v}_{1}^{k-2}}=\\ f_{\rs{v}_{k-1}}\int_{\mathbb{R}^m}f_{\rs{e}_{k}|\rs{e}_{k-1},\rs{e}_0,\rs{\delta}_0,\rs{v}_{1}^{k-1}}f_{\rs{e}_{k-1}|\rs{e}_{0},\rs{\delta}_0,\rs{v}_{1}^{k-2}}de_{k-1}\label{eq:secondSubErgo}, 
        \end{multline}which is proven in (\ref{eq:secondSubErgop1})-(\ref{eq:secondSubErgop2}) shown at the top of the subsequent page.
\begin{figure*}[!t]
\normalsize
\begin{IEEEeqnarray}{rCl}
        f_{\rs{e}_{k},\rs{v}_{k-1}|\rs{e}_0,\rs{\delta}_0,\rs{v}_{1}^{k-2}} &=& \int_{\mathbb{R}^m}f_{\rs{e}_{k}|\rs{e}_{k-1},\rs{e}_0,\rs{\delta}_0,\rs{v}_{1}^{k-1}}f_{\rs{v}_{k-1}|\rs{e}_{k-1},\rs{e}_0,\rs{\delta}_0,\rs{v}_{1}^{k-2}}f_{\rs{e}_{k-1}|\rs{e}_{0},\rs{\delta}_0,\rs{v}_{1}^{k-2}}de_{k-1}\label{eq:secondSubErgop1}\\&=& f_{\rs{v}_{k-1}}\int_{\mathbb{R}^m}f_{\rs{e}_{k}|\rs{e}_{k-1},\rs{e}_0,\rs{\delta}_0,\rs{v}_{1}^{k-1}}f_{\rs{e}_{k-1}|\rs{e}_{0},\rs{\delta}_0,\rs{v}_{1}^{k-2}}de_{k-1}\label{eq:secondSubErgop2}
\end{IEEEeqnarray} 
\hrulefill
\vspace*{4pt}
\end{figure*}
 Thus, substituting (\ref{eq:secondSubErgo}) into (\ref{eq:bayesthm}) and using the fact that  $\rvec{v}_{k-1}\indep (\rvec{e}_{0},\rvec{\delta}_{0},\rvec{v}^{k-2}_{1})$ we can write
\begin{multline}\label{eq:gconv}
        f_{\rs{e}_{k}|\rs{e}_0,\rs{\delta}_0,\rs{v}_{1}^{k-1}} =\\\int_{\mathbb{R}^{m}}f_{\rs{e}_{k}|\rs{e}_{k-1},\rs{e}_0,\rs{\delta}_0,\rs{v}_{1}^{k-1}}f_{\rs{e}_{k-1}|\rs{e}_{0},\rs{\delta}_0,\rs{v}_{1}^{k-2}}de_{k-1}.
\end{multline} From the recursion relationship (\ref{eq:algebraicMassage}) and  that $\rvec{w}_{k-1}\indep $ $(\rvec{e}_{k-1},\rvec{e}_{0},\rvec{\delta}_{0},\rvec{v}_{1}^{k-1})$, we have that 
\begin{multline}
    f_{\rs{e}_{k}|\rs{e}_{k-1},\rs{e}_0,\rs{\delta}_0,\rs{v}_{1}^{k-1}}({e}_{k}|{e}_{k-1},{e}_0,{\delta}_0,{v}_{1}^{k-1}) =\\   f_{\rs{e}_{k}|\rs{e}_{k-1},\rs{v}_{k-1}}({e}_{k}|{e}_{k-1},{v}_{k-1}) 
\end{multline} and thus, 
\begin{multline}
   f_{\rs{e}_{k}|\rs{e}_{k-1},\rs{e}_0,\rs{\delta}_0,\rs{v}_{1}^{k-1}}({e}_{k}|{e}_{k-1},{e}_0,{\delta}_0,{v}_{1}^{k-1}) =\\ N(e_{k};Re_{k-1}-Lv_{k-1},W).
\end{multline}Then, by the inductive assumption we have
\begin{multline}
f_{\rs{e}_{k-1}|\rs{e}_{0},\rs{\delta}_0,\rs{v}_{1}^{k-2}}({e}_{k-1}|{e}_{0},{\delta}_0,{v}_{1}^{k-2}) =\\ N(e_{k-1};\mu_{k-1}(e_0,{\delta}_0,v_1^{k-2}),\Sigma_{k-1}).
\end{multline} The integration in (\ref{eq:gconv}) is essentially a convolution of two  Gaussian PDFs, namely $ f_{\rs{e}_{k}|\rs{e}_0,\rs{\delta}_0,\rs{v}_{1}^{k-1}} = \int_{\mathbb{R}^{m}}N(e_{k};Re_{k-1}-Lv_{k-1},W)N(e_{k-1};\mu_{k-1}(e_0,{\delta}_0,v_1^{k-2}),\Sigma_{k-1})de_{k-1}.$ Computing this convolution gives
\begin{multline}
     f_{\rs{e}_{k}|\rs{e}_0,\rs{\delta}_0,\rs{v}_{1}^{k-1}} =\\ N(e_{k};R\mu_{k-1}(e_0,{\delta}_0,v_1^{k-2})-Lv_{k-1},R\Sigma_{k-1}R^{\tp}+W).\nonumber
\end{multline}
Substituting the assumed formulas (\ref{eq:normalStepFormula}) for $\Sigma_{k-1}$ and  $\mu_{k-1}(e_0,{\delta}_0,v_1^{k-2})$ into
\begin{align}
    \Sigma_k = R\Sigma_{k-1}R^{\tp}+W
\end{align} and
\begin{align}
\mu_k(e_0,{\delta}_0,v_1^{k-1}) = R\mu_{k-1}(e_0,{\delta}_0,v_1^{k-2})-Lv_{k-1}
\end{align}
 exactly recovers the formula (\ref{eq:normalStepFormula}) predicts for $n=k$. 
\end{proof} Before continuing, we will state and prove a lemma that describes some properties of the sequence of covariance matrices $\{\Sigma_{n}\}$ and the sequence of functions $\mu_{n}(e_{0},\delta_{0},v_{1}^{n-1}):\mathbb{D}^{m}\times ([-\frac{\Delta}{2},\frac{\Delta}{2}]^{m})^{n-1}\rightarrow \mathbb{R}^{m}$ described in Lemma \ref{lemm:normalPDF}. First, we recall a classic result from System Theory. Recall that for a matrix $X\in\mathbb{R}^{m\times m}$, we defined $\lVert X \rVert_{2}$ as the maximum singular value of $X$ and $\meig(X)$ as $X$'s spectral radius (the largest of the absolute values of $X$'s eigenvalues). 
\begin{proposition}[Gelfand's Theorem (cf. e.g. \cite{dullRobust}) and a Corollary]\label{prop:gelfand}
Gelfand's theorem states that $X\in \mathbb{R}^{m\times m}$, then 
\begin{align}\label{eq:gelfandguarantee}
    \lim_{n\rightarrow\infty }\left(\lVert X^n \rVert_{2} \right)^{\frac{1}{n}} = \meig(X). 
\end{align} If $\meig(X)<1$, then $\tau = (\meig(X)+1)/2$ has $\tau<1$. An immediate corollary of (\ref{eq:gelfandguarantee}) is that there exists $i\in\mathbb{N}_{+}$ such that for all $j\ge i$, $\lVert X^{j} \rVert_{2}\le \tau^j$. 
\end{proposition} In other words, Prop. \ref{prop:gelfand} guarantees that if $\meig(X)<1$, $\lim_{i\rightarrow\infty}X^{j}=0_{m\times m}$ ``geometrically fast". The next lemma concerns the sequence $\{\Sigma_{n}\}$.
\begin{lemma}\label{lemm:sigprops}
Let $\{\Sigma_{n}\}$ be the sequence of matrices in (\ref{eq:sigmandef}). For all $n$, we have $\Sigma_{n}\succeq W \succ 0_{m\times m}$. Furthermore, there exists a constant $c$ such that $\lVert\Sigma_{n}\rVert_{2}\le c$. 
\end{lemma}
\begin{proof}
It is immediate from (\ref{eq:sigmandef}) that $\Sigma_{n}\succeq W$. Note also that $\Sigma_{n}\succeq \Sigma_{n-1}$, so $\lVert \Sigma_{n}\rVert_{2}\ge  \lVert\Sigma_{n-1}\rVert_{2}$. Recall that $\meig(R)<1$, and let $\tau = (\meig(R)+1)/2$. By Prop. \ref{prop:gelfand}, there exists $j$ such that if $i\ge j$, $\lVert R^{i}\rVert_{2} \le \tau^{i}$. 
\begin{IEEEeqnarray}{rCl}
        \lVert \Sigma_{n}\rVert_{2} &=& \lVert \sum\limits_{i=0}^{n-1}R^iW(R^{\tp})^{i}\rVert_{2}\label{eq:thedefsign}\\ &\le& \lim_{n\rightarrow\infty } \sum\limits_{i=0}^{n-1}\lVert R^iW(R^{\tp})^{i}\rVert_{2}\label{eq:themonotonicityofsign}\\ &\le& \lim_{n\rightarrow\infty }\sum\limits_{i=0}^{n-1}\lVert R^i\rVert_{2}^2 \lVert W\rVert_{2}\label{eq:thesubmultiplicitivyofmn}\\ &\le& \sum\limits_{i=0}^{j-1}\lVert R^i\rVert_{2}^2 \lVert W\rVert_{2}+\lim_{n\rightarrow\infty }\sum\limits_{i=j}^{n}\tau^{2i} \lVert W\rVert_{2}\label{eq:weusegelfand}\\ &\le&  \sum\limits_{i=0}^{j-1}\lVert R^i\rVert_{2}^2 \lVert W\rVert_{2}+\lVert W\rVert_{2}\frac{1}{1-\tau^2}\label{eq:geometrictau},
\end{IEEEeqnarray} where (\ref{eq:thedefsign}) is the definition  (\ref{eq:sigmandef}), (\ref{eq:themonotonicityofsign}) follows from the triangle inequality and monotonicity, (\ref{eq:thesubmultiplicitivyofmn}) is from the fact that the matrix norm $\lVert\circ \rVert_{2}$ is submultiplicative, (\ref{eq:weusegelfand}) applies the corollary in Prop. \ref{prop:gelfand}, and finally (\ref{eq:geometrictau}) is the geometric series formula (note $\tau<1$).  Making the choice $c =\sum\limits_{i=0}^{j-1}\lVert R^i\rVert_{2}^2 \lVert W\rVert_{2}+\lVert W\rVert_{2}\frac{1}{1-\tau^2}$ proves the result. 
\end{proof}
The next lemma concerns the sequence of functions $\mu_{n}$ in (\ref{eq:mundef}). Namely, it proves that the range of the functions lies in compact set that does not depend on $n$ or the realizations $v_{1}^{n-1}$. 
\begin{lemma}\label{lemm:muprops}
There exists constants $\alpha$ and $\beta$ such that for any $n$ and choice of $v_{1}^{n-1}\in ([-\frac{\Delta}{2},\frac{\Delta}{2}]^m)^{n-1}$ we have $\lVert \mu_n(e_0,{\delta}_0,v_1^{n-1}) \rVert_{2}\le \alpha \lVert M(e_{0},\delta_{0}) \rVert_{2} +\beta$.
\end{lemma}
\begin{proof}
The proof is analogous to Lemma \ref{lemm:sigprops}. Let $\alpha =  \max_{n\in\mathbb{N}_{+}} \lVert R^{n-1}\rVert_{2}$. Since $\meig(R)<1$, we have $\alpha<\infty$ by the corollary in Prop. \ref{prop:gelfand}. Let $\tau =(\meig(R)+1)/2$ and let $j$ be as in the statement of Prop. \ref{prop:gelfand}. The proof follows from the inequalities (\ref{eq:trisub})-(\ref{eq:usecorr2}), illustrated at the top of the following page.
\begin{figure*}[!t]
\normalsize
\begin{IEEEeqnarray}{rCl}
    \lVert\mu_n(e_0,{\delta}_0,v_1^{n-1}) \rVert_{2} &\le& \lVert R^{n-1}\rVert_{2}  \lVert M(e_{0},\delta_{0})\rVert_{2}+\sum_{i=0}^{n-1}\lVert R^i\rVert_{2}\lVert L \rVert_{2}\lVert v_{n-1-i}\rVert_{2}\label{eq:trisub}\\ &\le& \alpha \lVert M(e_{0},\delta_{0})\rVert_{2}+\sum_{i=0}^{n-1}\lVert R^i\rVert_{2}\lVert L \rVert_{2}\frac{\sqrt{m}}{2}\Delta\label{eq:extremal},\\&\le& \alpha \lVert M(e_{0},\delta_{0})\rVert_{2}+\lim_{n\rightarrow\infty}\sum_{i=0}^{n}\lVert R^i\rVert_{2}\lVert L \rVert_{2}\frac{\sqrt{m}}{2}\Delta \\&\le& \alpha \lVert M(e_{0},\delta_{0})\rVert_{2}+\sum_{i=0}^{j-1}\lVert R^i\rVert_{2}\lVert L \rVert_{2}\frac{\sqrt{m}}{2}\Delta + \lim_{n\rightarrow\infty }\sum_{i=j}^{n}\tau^i\lVert L \rVert_{2}\frac{\sqrt{m}}{2}\Delta\label{eq:usecorr1}\\&\le& \alpha \lVert M(e_{0},\delta_{0})\rVert_{2}+\sum_{i=0}^{j-1}\lVert R^i\rVert_{2}\lVert L \rVert_{2}\frac{\sqrt{m}}{2}\Delta + \frac{1}{1-\tau}\lVert L \rVert_{2}\frac{\sqrt{m}}{2}\Delta\label{eq:usecorr2}
\end{IEEEeqnarray}
\hrulefill
\vspace*{4pt}
\end{figure*}
For all $n$, (\ref{eq:trisub}) follows immediately from the triangle inequality and submultiplicativity. Then, (\ref{eq:extremal}) follows by the definition of $\alpha$ and the fact that since the $v_{i}\in[-\frac{\Delta}{2},\frac{\Delta}{2}]^m$, they have $\lVert v_{i}\rVert_{2} \le \sqrt{m}\Delta/2$. Finally (\ref{eq:usecorr1}) and (\ref{eq:usecorr2}) are completely analogous to  (\ref{eq:weusegelfand}) and (\ref{eq:geometrictau}). 
\end{proof}
With these in hand, we prove that the Markov process $\{\rs{e}_{t},\rs{\delta}_{t}\}$ satisfies the hypothesis of Theorem \ref{thm:existanceCrit}.
\begin{lemma}\label{lemm:transientmeasure0}
All sets $F\in\mathbb{B}(\mathbb{D}^{m})$ that are weakly transient with respect to the Markov kernel (\ref{eq:transitionpdf}) have  $\lambda(F)=0$. 
\end{lemma}
\begin{proof}
 We proceed via the contrapositive. Namely, we demonstrate that if  $F\in\mathbb{B}(\mathbb{D}^{m})$ has $\lambda(F)>0$, then $F$ is not weakly transient with respect to the Markov kernel (\ref{eq:transitionpdf}). 
Assume that $F\in\mathbb{B}(\mathbb{D}^{m})$ has $\lambda(F)>0$. We will prove that for every such $F$ and initial condition $e_{0},{\delta}_{0}$ there exists $\xi>0$ such that
\begin{align}\label{eq:whatwesetouttoprove}
    \mathbb{P}_{\rs{e}_{n},\rs{\delta}_{n}|\rs{e}_{0},\rs{\delta}_{0}}[F|\rs{e}_{0}={e}_{0},\rs{\delta}_{0}={\delta}_{0}] > \xi\textit{, for all $n\in\mathbb{N}_+$}.
\end{align} This ensures that for every  $e_{0},{\delta}_{0}$ and subsequence $n_{i}\in\mathbb{N}$
\begin{align}
    \sum_{i=1}^{\infty}\mathbb{P}_{\rs{e}_{n_{i}},\rs{\delta}_{n_i}|\rs{e}_{0},\rs{\delta}_{0}}[F|\rs{e}_{0}={e}_{0},\rs{\delta}_{0}={\delta}_{0}]=\infty.
\end{align} 

Since $F\subset\mathbb{D}^{m}$ has positive Lebesgue measure, the regularity of Lebesgue measure (cf. \cite[Thms. 2.14, 2.18]{rudinrc}) implies that $F$ must contain a compact set with strictly positive Lebesgue measure; in other words, there exists a closed, bounded $H\subset F$ such that for some $\kappa > 0$, $\lambda(H) = \kappa$. By countable additivity for all $n\in\mathbb{N}_+$
\begin{multline}\label{eq:countableadditivity}
     \mathbb{P}_{\rs{e}_{n},\rs{\delta}_{n}|\rs{e}_{0},\rs{\delta}_{0}}[(\rs{e}_{n},\rs{\delta}_{n})\in F|\rs{e}_{0}={e}_{0},\rs{\delta}_{0}={\delta}_{0}]\ge\\ \mathbb{P}_{\rs{e}_{n},\rs{\delta}_{n}|\rs{e}_{0},\rs{\delta}_{0}}[(\rs{e}_{n},\rs{\delta}_{n})\in H|\rs{e}_{0}={e}_{0},\rs{\delta}_{0}={\delta}_{0}] .
\end{multline} Consider a fixed $n\in\mathbb{N}_{+}$. It is obvious that 
\begin{multline}\label{eq:rectangleBound}
    \mathbb{P}_{\rs{e}_{n},\rs{\delta}_{n}|\rs{e}_{0},\rs{\delta}_{0}}[(\rs{e}_{n},\rs{\delta}_{n})\in H|\rs{e}_{0}={e}_{0},\rs{\delta}_{0}={\delta}_{0}]\ge\\ \kappa\inf_{(x,y)\in H}f_{n|0}(x,y|\rs{e}_{0}=e_{0},\rs{\delta}_0 = {\delta}_{0}). 
\end{multline} We establish the result of the lemma by finding a lower bound for the infimum on the right-hand side of (\ref{eq:rectangleBound})
that does not depend on $n$. Let $H_{e}=\{x\in\mathbb{R}^{m}: \exists \delta \in [-{\Delta}/2,{\Delta}/2]^{{m}} \text{ with } (x,\delta)\in H \}$ denote the ``$e-$section" of $H$. Note that  $H_{e}$ is a compact subset of $\mathbb{R}^{m}$. Boundedness of $H_{e}$ is inherited from the boundedness of $H$. To see that $H_{e}$ is closed, let $x$ be a limit point of $H_{e}$ and the limit of the sequence $x_{i}\in H_{e}$. Since $x_{i}\in H_{e}$, for each $x_{i}$ there exists a $\delta_{i}$ such that $(x_{i},\delta_{i})\in H$. Since $H$ is compact, a subsequence of $(x_{i},\delta_{i})$, denoted $(x_{n_{i}},\delta_{n_{i}})$, converges in $H$; in other words, there exists some $\overline{\delta}$ such that $\lim_{i}(x_{n_{i}},\delta_{n_{i}}) = (x,\overline{\delta})$ and $(x,\overline{\delta})\in H$. Since $(x,\overline{\delta})\in H$, $x\in H_{e}$. Since $H_{e}$ contains its limit points, it is closed.
From the factorization of the n-step transition PDF (\ref{eq:nStepTransition}) and the fact that $H$ is contained strictly inside $\mathbb{D}^{m}$ we have
\begin{align}\label{eq:secondinchain}
      \inf_{(x,y)\in H}f_{n|0}(x,y|e_{0},{\delta}_{0}) =     \inf_{x\in H_e}\frac{f_{\rs{e}_{n}|\rs{e}_{0},\rs{\delta}_0}(x|e_{0},{\delta}_{0})}{\Delta^m}.
\end{align} 
By definition,
\begin{multline}\label{eq:theplan}
    f_{\rs{e}_{n}|\rs{e}_{0},\rs{\delta}_0}(x|e_{0},{\delta}_{0}) =\\ \mathbb{E}_{\rs{v}_{1}^{n-1}|\rs{e}_{0}=e_0,\rs{\delta}_0={\delta}_0}[f_{\rs{e}_{n}|\rs{v}_{1}^{n-1},\rs{e}_{0},\rs{\delta}_0}(x|\rs{v}_{1}^{n-1},e_{0},{\delta}_{0})].
\end{multline} 
Recall the reconstruction error satisfies $\rs{v}_i\in[-\frac{\Delta}{2},\frac{\Delta}{2}]^m$ for all $i$. Since ``the minimum is less than or equal to the average", we have
\begin{multline}\label{eq:minleav}
    \inf_{x\in H_e}f_{\rs{e}_{n}|\rs{e}_{0},\rs{\delta}_0}(x|e_{0},{\delta}_{0})\ge\\\inf_{\substack{x\in H_e\\ v_{1}^{n-1}\in {([-\frac{\Delta}{2},\frac{\Delta}{2}]^m)}^{(n-1)}}}f_{\rs{e}_{n}|\rs{v}_{1}^{n-1},\rs{e}_{0},\rs{\delta}_0}(x|{v}_{1}^{n-1},e_{0},{\delta}_{0})
\end{multline} In Lemma \ref{lemm:normalPDF}'s (\ref{eq:normalStepFormula}) we demonstrated that for any $n$, realizations of the reconstruction error ${v}_1^{n-1}$, and realizations of the initial conditions $\rs{e}_{0},\rs{\delta}_0$ we have
\begin{multline}\label{eq:unlabeleduptonow}
    f_{\rs{e}_{n}|\rs{v}_{1}^{n-1},\rs{e}_{0},\rs{\delta}_0}(x|{v}_{1}^{n-1},e_{0},{\delta}_{0}) =\\ N(x;\mu_{n}(e_{0},\delta_{0},v_{1}^{n-1}),\Sigma_{n} ).
\end{multline} Note that for $\Sigma \succeq W \succ 0$ and $\mu,x\in\mathbb{R}^{m}$, the function $N(x;\mu,\Sigma)$ is strictly positive and continuous in $(x,\mu,\Sigma)$. Let $c$ be as in the statement of Lemma \ref{lemm:sigprops} and let $\alpha$ and $\beta$ be as in the statement of Lemma \ref{lemm:muprops}. 
Define the set $\mathcal{M}(e_{0},\delta_{0},m,L,R,\Delta)\subset \mathbb{R}^{m}\times \mathbb{S}^{m\times m}_{+} $ via $\mathcal{M}(e_{0},\delta_{0},m,L,R,\Delta)=\{\mu\in\mathbb{R}^m,\Sigma\in\mathbb{S}^{m\times m}_{+}: \Sigma\succeq W, \lVert\Sigma\rVert_{2}\le c,  \lVert\mu\rVert_{2}\le \alpha \lVert M(e_{0},\delta_{0})\rVert_{2}+\beta \}$. This set is compact (closed and bounded). For any $n$ and $v_{1}^{n-1}\in {([-\frac{\Delta}{2},\frac{\Delta}{2}]^m)}^{n-1}$, Lemmas \ref{lemm:sigprops}  and \ref{lemm:muprops} guarantee that we have that $\left(\mu_{n}(e_0,\delta_{0},v_{1}^{n-1}), \Sigma_{n}\right) \in  \mathcal{M}(e_{0},\delta_{0},m,L,R,\Delta)$. 

Via (\ref{eq:unlabeleduptonow}) we have
\begin{multline}
    \inf_{\substack{x\in H_e\\ v_{1}^{n-1}\in {([-\frac{\Delta}{2},\frac{\Delta}{2}]^m)}^{(n-1)}}}f_{\rs{e}_{n}|\rs{v}_{1}^{n-1},\rs{e}_{0},\rs{\delta}_0}(x|{v}_{1}^{n-1},e_{0},{\delta}_{0})= \\ \inf_{\substack{x\in H_e\\ v_{1}^{n-1}\in {([-\frac{\Delta}{2},\frac{\Delta}{2}]^m)}^{(n-1)}}} N(x;\mu_{n}(e_{0},\delta_{0},v_{1}^{n-1}),\Sigma_{n} ).\label{eq:refmeatendchainequality}
    \end{multline} As for any choice of $v_{1}^{n-1}\in {([-\frac{\Delta}{2},\frac{\Delta}{2}]^m)}^{(n-1)}$ we have $\left(\mu_{n}(e_0,\delta_{0},v_{1}^{n-1}), \Sigma_{n}\right) \in  \mathcal{M}(e_{0},\delta_{0},m,L,R,\Delta)$ gives that
    \begin{multline}
    \inf_{\substack{x\in H_e\\ v_{1}^{n-1}\in {([-\frac{\Delta}{2},\frac{\Delta}{2}]^m)}^{(n-1)}}} N(x;\mu_{n}(e_{0},\delta_{0},v_{1}^{n-1}),\Sigma_{n} ) \ge \\ \inf_{\substack{x\in H_e\\ (\mu,\Sigma)\in \mathcal{M}(e_{0},\delta_{0},m,L,R,\Delta) } }N(x;\mu,\Sigma),\label{eq:indylb}
\end{multline}where we note that the lower bound in (\ref{eq:indylb}) does not depend on $n$ or $v_{1}^{n-1}$. Furthermore (\ref{eq:indylb}) is a minimization of a strictly positive function over the compact (closed and bounded set) given by $\mathcal{C}=\{x\in\mathbb{R}^{m},\mu\in\mathbb{R}^{m},\Sigma\in\mathbb{S}_{+}^{m}:x\in H_{e}, (\mu,\Sigma)\in \mathcal{M}(e_{0},\delta_{0},m,L,R,\Delta) \}$. The function minimized, $N(x;\mu,\Sigma)$ is continuous on $\mathcal{C}$ since $\Sigma \succ 0_{m\times m}$. Thus, for some $\epsilon>0$
\begin{align}\label{eq:epspos}
\inf_{\substack{x\in H_e\\ (\mu,\Sigma)\in \mathcal{M}(e_{0},\delta_{0},m,L,R,\Delta) } }N(x;\mu,\Sigma) \ge \epsilon.
\end{align}  Connecting the chain of inequalities (\ref{eq:countableadditivity}), (\ref{eq:rectangleBound}), (\ref{eq:secondinchain}), (\ref{eq:minleav}), (\ref{eq:refmeatendchainequality}), (\ref{eq:indylb}), and (\ref{eq:epspos})  gives:
\begin{subequations}
\begin{multline}
    \mathbb{P}_{\rs{e}_{n},\rs{\delta}_{n}|\rs{e}_{0},\rs{\delta}_{0}}[(\rs{e}_{n},\rs{\delta}_{n})\in F|\rs{e}_{0}={e}_{0},\rs{\delta}_{0}={\delta}_{0}] \ge \\  \mathbb{P}_{\rs{e}_{n},\rs{\delta}_{n}|\rs{e}_{0},\rs{\delta}_{0}}[(\rs{e}_{n},\rs{\delta}_{n})\in H|\rs{e}_{0}={e}_{0},\rs{\delta}_{0}={\delta}_{0}]
    \end{multline} and, finally
    \begin{multline}
    \mathbb{P}_{\rs{e}_{n},\rs{\delta}_{n}|\rs{e}_{0},\rs{\delta}_{0}}[(\rs{e}_{n},\rs{\delta}_{n})\in H|\rs{e}_{0}={e}_{0},\rs{\delta}_{0}={\delta}_{0}] \ge \frac{\kappa}{\Delta^m}\epsilon.  \label{eq:rectangleBound2}
\end{multline}  
\end{subequations}
Thus, choosing $\xi = \frac{\kappa}{\Delta^m}\epsilon$ establishes (\ref{eq:whatwesetouttoprove}). 
Thus, if $F$ has $\lambda(F)>0$, for any subsequence $\{n_i\}\in\mathbb{N}$ and $e_{0},\delta_{0}$ there exists $\xi>0$ such that 
\begin{align}\label{eq:finalseries}
 \sum_{i=1}^{r}\mathbb{P}_{\rs{e}_{n_{i}},\rs{\delta}_{n_i}|\rs{e}_{0},\rs{\delta}_{0}}[(\rs{e}_{n_{i}},\rs{\delta}_{n_i})\in F|\rs{e}_{0}={e}_{0},\rs{\delta}_{0}={\delta}_{0}]\ge  r\xi .
\end{align} The series thus diverges as $r\rightarrow \infty$ for any initial condition $e_{0},\delta_{0}$. This implies that $F$ is not weakly transient per Definition \ref{def:weakTrans}, thus all weakly transient sets have Lebesgue measure $0$.
\end{proof}
Theorem \ref{thm:existanceCrit} guarantees that an invariant measure that is equivalent to $\lambda$ exists if all weakly transient sets have $\lambda$-measure 0. Combining Lemma \ref{lemm:transientmeasure0} with Theorem \ref{thm:existanceCrit}  proves Lemma \ref{lemm:invarexist}. 
\subsection{Proof of Lemma \ref{lemm:ergo}}
We now demonstrate that the Markov chain describing (jointly) the dither and innovation processes satisfies some ergodic properties; in particular that the sequence of random variables $(\rvec{e}_{t},\rvec{\delta}_{t})$ converge in distribution to the invariant measure. We begin again with some definitions and a key result from the survey \cite{mcmcReview}. 

\begin{definition}[\cite{mcmcReview}]\label{def:pi}
A Markov chain $\{\rs{z}_i\}$ on some state space $\mathbb{X}$ is called $\phi$-\textit{irreducible} if there exists a nonzero $\sigma$-finite measure $\phi$ such that for all measurable $\mathcal{A}\subset \mathbb{X}$ with $\phi(\mathcal{A})>0$ and all initial conditions $\rs{z}_{0}=z_{0}$ with $z_0\in\mathbb{X}$ we can find an integer $n$ such that 
\begin{align}
    \mathbb{P}_{\rs{z}_{n}|\rs{z}_{0}}[\rs{z}_{n}\in\mathcal{A}|\rs{z}_{0}= {z}_{0}] > 0.
\end{align}
\end{definition} 
\begin{definition}[\cite{mcmcReview}]\label{def:ap}
A Markov chain on $\mathbb{X}$ is called aperiodic if there does not exist $d> 1$ and disjoint nonempty measurable subsets $\mathcal{Z}_{0},\mathcal{Z}_{1}, \dots \mathcal{Z}_{d-1}$ such that when ${z}_{n-1}\in\mathcal{Z}_{i}$
\begin{align}
    \mathbb{P}_{\rs{z}_{n}|\rs{z}_{n-1}}[\rs{z}_{n}\in\mathcal{Z}_{i+1\text{ mod } d }|\rs{z}_{n-1}={z}_{n-1}] = 1.
\end{align} 
\end{definition}
\begin{definition}[Total Variation]
Define the total variation norm between two probability measures $\mathbb{P}_{1}:\mathbb{B}(\mathbb{X})\rightarrow[0,1]$ and $\mathbb{P}_{2}:\mathbb{B}(\mathbb{X})\rightarrow[0,1]$ defined on the same measure space via
\begin{align}
    \lVert \mathbb{P}_{1} - \mathbb{P}_{2} \rVert_{\mathrm{T.V.}} \overset{\Delta}{=}     \sup_{\mathcal{A}\in\mathbb{\mathbb{X}}}|\mathbb{P}_1(\mathcal{A})-\mathbb{P}_2(\mathcal{A})|.
\end{align} 
\end{definition}
\begin{theorem}[{\cite[Theorem 4]{mcmcReview}}]\label{thm:converg}
Consider a Markov chain $\{\rs{r}_{i}\}$ on a countably generated state space that is aperiodic, $\phi$-irreducible, and admits an invariant measure $\mathbb{P}_{\mathrm{inv}}$ that is absolutely continuous with respect to Lebesgue measure. 
For $\lambda$-almost every initial condition ${r}_{0}$ we have 
\begin{align}
    \lim_{n\rightarrow \infty }\lVert \mathbb{P}_{\rs{r}_{n}|\rs{r}_{0}}[\rs{r}_{n}\in \circ | \rs{r}_{0}={r}_{0}]-\mathbb{P}_{\mathrm{inv}}(\circ) \rVert_{\mathrm{T.V.}} =0.
    \end{align} Furthermore, the law of large numbers holds in the following sense. Assume the initial state of the chain $\rs{z}_0$ is a random variable that is absolutely continuous with respect to $\lambda$. For all measurable functions $\eta$ such that $\mathbb{E}_{\rs{r}\sim \mathbb{P}_{\mathrm{inv}}}[|\eta(\rs{r})|]<\infty$ and 
     $\mathbb{E}_{\rs{r}_0}[|\eta(\rs{r}_0)|]<\infty$
    we have
    \begin{align}\label{eq:lln_def}
        \lim_{N\rightarrow \infty }\frac{1}{N+1}\sum\limits_{i=0}^{N} \eta(\rs{r}_{i}) \overset{\mathrm{a.s.}}{=} \mathbb{E}_{\rs{r}\sim\mathbb{P}_{\mathrm{inv}}}[\eta(\rs{r})].
    \end{align}
\end{theorem}  In the present setting, the state space $\mathbb{X}=\mathbb{D}^m$. The Borel $\sigma$ - algebra on $\mathbb{D}^{m}$ is countably generated, and the Lebesgue measure  on $\mathbb{D}^{m}$ (denoted  $\lambda$) is $\sigma-$finite. Thus, to guarantee that the n-step conditional probability measures for the Markov chain $\{\rs{e}_{t},\rs{\delta}_{t}\}$ defined by (\ref{eq:transitionpdf}) will converge to the stationary distribution in total variation, and to verify that the law of large numbers holds in the sense of (\ref{eq:lln_def}), we can verify that the chain is $\lambda$-irreducible and aperiodic.
\begin{lemma}\label{lemm:tv_conv_hy}
The Markov chain induced by (\ref{eq:transitionpdf}) and (\ref{eq:transitionpdf}) is $\lambda$-irreducible and aperiodic.
\end{lemma}

\begin{proof}
We first demonstrate $\lambda-$irreducibility. Let $\mathcal{A}\subset\mathbb{D}^{m}$ be any set of positive Lebesgue measure. Take $(e_0,\delta_{0})\in\mathbb{D}^{m}$. We have, by (\ref{eq:transitionpdf}),
\begin{multline}
    \mathbb{P}_{\rs{e}_{1},\rs{\delta}_{1}|\rs{e}_0,\rs{\delta}_0}[(\rs{e}_1,\rs{\delta}_1)\in\mathcal{A}|\rs{e}_0=e_0,\rs{\delta}_0=\delta_{0}] =\\ \iint_{\mathcal{A}} f_{t+1|t}(e_1,\delta_1|e_0,\delta_0)de_1d\delta_1,\end{multline}
    \begin{multline}
        \iint_{\mathcal{A}} f_{t+1|t}(e_1,\delta_1|e_0,\delta_0)de_1d\delta_1 =\\ \iint_{\mathcal{A}} \frac{1}{\Delta^m}\frac{e^{-\frac{1}{2}(e_{1}-M(e_{0},{\delta}_{0}))^{\mathrm{T}}W^{-1}(e_{1}-M(e_{0},{\delta}_{0}))}}{\sqrt{(2\pi)^m \det{W} }}de_1d\delta_1,
\end{multline} and, finally, \begin{align}
 \iint_{\mathcal{A}} \frac{1}{\Delta^m}\frac{e^{-\frac{1}{2}(e_{1}-M(e_{0},{\delta}_{0}))^{\mathrm{T}}W^{-1}(e_{1}-M(e_{0},{\delta}_{0}))}}{\sqrt{(2\pi)^m \det{W} }}de_1d\delta_1 >0.  \nonumber
\end{align} This established that taking $n=1$ always allows us to satisfy the requirements for $\lambda$-irreducibility. This allows the proof of aperiodicity to follow immediately. We proceed by contradiction. Assume the chain is periodic (i.e., assume that the chain is ``not aperiodic" via Definition \ref{def:ap}); assume that one has a set of $d>1$ disjoint nonempty measurable subsets $\mathcal{S}_{0},\mathcal{S}_{1},\dots,\mathcal{S}_{d-1}\subset \mathbb{D}^{m}$ such that for all $t$ and $i$ when $(e_{t},\delta_{t})\in\mathcal{S}_{i}$, $\mathbb{P}_{\mathbf{e}_{t+1}|\mathbf{e}_{t}}[(\mathbf{e}_{t+1},\boldsymbol{\delta}_{t+1})\in\mathcal{S}_{i+1 \mod d}|\mathbf{e}_{t}=e_{t},\boldsymbol{\delta}_{t}=\delta_{t}] = 1$. Take $(e_t,{\delta}_t)\in\mathcal{S}_{0}$. By assumption
\begin{align}
    \mathbb{P}[(\rs{e}_{t+1},\rs{\delta}_{t+1})\in\mathcal{S}_{1}|\rs{e}_t=e_t, \rs{\delta}_t={\delta}_t] = 1. 
\end{align} Note that by (\ref{eq:transitionpdf}), $\lambda(\mathcal{S}_{1})>0$ or else  $\mathbb{P}[(\rs{e}_{t+1},\rs{\delta}_{t+1})\in\mathcal{S}_{1}|\rs{e}_t=e_t, \rs{\delta}_t={\delta}_t]=0$. By our work proving the irreducibly condition, it must be that $\mathcal{S}_{1}=\mathbb{R}$, i.e., the whole state space. This is a contradiction, since the hypothesis of Definition \ref{def:ap} is that $\mathcal{S}_{0}$ is nonempty and $\mathcal{S}_{0}\cap\mathcal{S}_{1}=\emptyset$.
\end{proof}
Lemma \ref{lemm:tv_conv_hy} verifies the hypothesis of Theorem \ref{thm:converg} and thus proves Lemma \ref{lemm:ergo}.  Note that the convergence in total variation implies weak convergence (i.e., convergence in distribution).

\end{document}